\documentclass[12pt]{amsart}   % reqno number on RHS of equations.
\usepackage{amscd,amssymb,cite,amsmath,latexsym,comment,epic,eepic,euscript}
\usepackage{graphicx}
\usepackage{enumerate}
\usepackage[initials]{amsrefs}
\usepackage[all]{xy}
\usepackage{tikz}
\usetikzlibrary{decorations.markings,arrows, decorations.pathreplacing}

\theoremstyle{plain}
\newtheorem{thm}{Theorem}[section]
\newtheorem{cor}[thm]{Corollary} %%Delete [thm] to re-start numbering
\newtheorem{lemma}[thm]{Lemma} %%Delete [thm] to re-start numbering
\newtheorem{prop}[thm]{Proposition}

\theoremstyle{remark}

\theoremstyle{definition}
\newtheorem{defi}[thm]{Definition}
\newtheorem{example}[thm]{Example}

\newtheorem{notation}[thm]{Notation}

\newtheorem{remark}[thm]{Remark}
\newtheorem{remarks}[thm]{Remarks}

\usepackage{amscd,amssymb,comment,epic,eepic,euscript}
\usepackage{graphicx}
\usepackage{enumerate}
\usepackage[initials]{amsrefs}
\usepackage[all]{xy}
\newcount\theTime
\newcount\theHour
\newcount\theMinute
\newcount\theMinuteTens
\newcount\theScratch
\theTime=\number\time
\theHour=\theTime
\divide\theHour by 60
\theScratch=\theHour
\multiply\theScratch by 60
\theMinute=\theTime
\advance\theMinute by -\theScratch
\theMinuteTens=\theMinute
\divide\theMinuteTens by 10
\theScratch=\theMinuteTens
\multiply\theScratch by 10
\advance\theMinute by -\theScratch

\def\today{{\number\day\space
 \ifcase\month\or
  January\or February\or March\or April\or May\or June\or
  July\or August\or September\or October\or November\or December\fi
 \space\number\year}}

\newcommand{\beq}{\begin{equation}}
\newcommand{\eeq}{\end{equation}}
\newcommand{\bre}{\begin{remark}}
\newcommand{\ere}{\end{remark}}
\newcommand{\beqno}[1]{\begin{equation}\label{#1}}

\newcommand{\mref}[1]{(\ref{#1})}
\newcommand\bZ{\mathbb Z}
\DeclareMathOperator{\supp}{supp}
\DeclareMathOperator{\Hess}{Hess}

\newenvironment{prmainKKR}{\paragraph{\textit{Proof of Theorem 3.1}}}{\hfill$\square$}
\newenvironment{prmainKR}{\paragraph{\textit{Proof of Theorem \ref{main_KR}}}}{\hfill$\square$}

\newenvironment{prmicrolocal}{\paragraph{\textit{Proof of Theorem \ref{microlocal_2}}}}{\hfill$\square$}
\newenvironment{prmicrolocalprop}{\paragraph{\textit{Proof of Proposition \ref{kernel_estimate}}}}{\hfill$\square$}
\newenvironment{praddfunc}{\paragraph{\textit{Proof of Proposition \ref{P:add_func}}}}{\hfill$\square$}
\newenvironment{prA_h}{\paragraph{\textit{Proof of Proposition \ref{A_h}}}}{\hfill$\square$}
\begin{document}

\title[Green's function asymptotics]{Green's function asymptotics of periodic elliptic operators on abelian coverings of compact manifolds}

\author[Kha]{Minh Kha}
\address{M.K., Department of Mathematics, Texas A\&M University,
College Station, TX 77843-3368, USA}
\email{kha@math.tamu.edu}
%\date{\timeanddate}
\begin{abstract}
The main results of this article provide asymptotics at infinity of the Green's functions near and at the spectral gap edges for ``generic" periodic second-order, self-adjoint, elliptic operators on noncompact Riemannian co-compact coverings with abelian deck groups. Previously, analogous results have been known for the case of $\mathbb{R}^n$ only. One of the interesting features discovered is that the rank of the deck group plays more important role than the dimension of the manifold.
\end{abstract}

\maketitle
%%%%%%%%%%%%%%%%%%%%%%%%%%
\section{Introduction}
%%%%%%%%%%%%%%%%%%%%%%%%%%
The behavior at infinity of the Green's function of the Laplacian $-\Delta$ on an Euclidean space below and at the boundary of the spectrum is well known. The main term of the asymptotics for the Green's function of any bounded below periodic second-order elliptic operator below and at the bottom of the spectrum was found in \cite{Bab,MT} (see also \cite{Woess} for discrete setting).
For such operators, the band-gap structure of their spectra is known (e.g., \cite{Ea, RS4}), and thus, spectral gaps may exist. Thus, it is interesting to derive the behavior of the Green's functions inside and at the edges of these gaps.
Recently, the corresponding results for a ``generic" periodic elliptic operator in $\mathbb{R}^n$
were established in \cite{KKR, KR}.

Meanwhile, many classical properties of solutions of periodic Schr\"odinger operators on Euclidean spaces were generalized successfully to solutions of periodic Schr\"odinger operators on coverings of compact manifolds (see e.g., \cite{Ag2, BruSun1, BruSun2, KOS, KP2, Sunada1, Sunada2, LinPinchover}).
Hence, a question arises of whether one can obtain 
analogs of the results of \cite{KKR, KR} as well.
The main theorems \ref{main} and \ref{main_KR} of this article provide such results for periodic operators on an abelian covering of a compact Riemannian manifold. The results are in line with Gromov's idea that the large scale geometry of a co-compact normal covering is captured mostly by its deck transformation group (see e.g., \cite{Cha, Gro, Saloff-Coste}). For instance, the dimension of the covering manifold does not enter explicitly to the asymptotics. Rather, the torsion-free rank $d$ of the abelian deck transformation group influences these asymptotics significantly.
One can find a similar effect in various
results involving analysis on Riemannian co-compact normal coverings such as 
the long time asymptotic behaviors of the heat kernel on a noncompact abelian Riemannian covering \cite{KoSu}, and the analogs of Liouville's theorem \cite{KP2} (see also \cite{Saloff-Coste} for an excellent survey on analysis on co-compact coverings).

We discuss now the main thrust of this paper.

Let $X$ be a noncompact Riemannian manifold that is a normal abelian covering of a compact Riemannian manifold $M$ with the deck transformation group $G$. For any function $u$ on $X$ and any $g \in G$, we denote by $u^g$ the ``shifted" function
$$u^g(x)=u(g\cdot x),$$
for any $x \in X$.
Consider a bounded below second-order, symmetric elliptic operator $L$ on the manifold $X$ with smooth coefficients. We assume that $L$ is a \textbf{periodic} operator on $X$, i.e., the following invariance condition holds:
$$Lu^{g}=(Lu)^{g},$$
for any $g\in G$ and $u \in C^{\infty}_c(X)$. The operator $L$, with the Sobolev space $H^2(X)$ as its domain, is an unbounded self-adjoint operator in $L^2(X)$. 

The following result for such operators is well-known (see e.g., \cite{BruSun1, Ea, K, RS4, Sunada1, Sunada2}):
\begin{thm}
\label{bandgap}
The spectrum of the above operator $L$ in $L^{2}(X)$ has a \textbf{band-gap structure}, i.e., it is the union of a sequence of closed bounded intervals $[\alpha_j, \beta_j] \subset \mathbb{R}$ $(j=1,2,...)$ (\textbf{bands} or \textbf{stability zones} of the operator $L$):
\begin{equation} 
\label{band-gap}
\sigma(L)=\bigcup_{j=1}^{\infty}[\alpha_j, \beta_j],
\end{equation}
such that $\alpha_{j} \leq \alpha_{j+1}$, $\beta_{j} \leq \beta_{j+1}$ and $\lim_{j \rightarrow \infty}\alpha_j=\infty$. 
\end{thm}

The bands can overlap when the dimension of the covering $X$ is greater than $1$, but they can leave open intervals in between, called \textbf{spectral gaps}. 
\begin{defi}
\label{gap}
A \textbf{finite spectral gap} is of the form $(\beta_{j}, \alpha_{j+1})$ for some $j \in \mathbb{N}$ such that $\alpha_{j+1}>\beta_{j}$, and the \textbf{semifinite spectral gap} is the open interval $(-\infty, \alpha_{1})$, which contains all real numbers below the bottom of the spectrum of $L$. 
\end{defi}
In this text, we study Green's function asymptotics for the operator $L$ at an energy level $\lambda \in \mathbb{R}$,
such that $\lambda$ belongs to the union of all closures of finite spectral gaps\footnote{All of the results still hold for the case when $\lambda$ does not exceed the bottom of the spectrum, i.e. for the semi-infinite gap.}.
We divide this into two cases:
\begin{itemize}
\item
\textbf{Case I:} (\textit{Spectral gap interior}) The level $\lambda$ is in a \textbf{finite} spectral gap $(\beta_j, \alpha_{j+1})$ such that $\lambda$ is close either to the spectral edge $\beta_j$ or to the spectral edge $\alpha_{j+1}$. 

\item 
\textbf{Case II:} (\textit{Spectral edge case}) The level $\lambda$ coincides with one of the spectral edges of some \textbf{finite} spectral gap, i.e., $\lambda=\alpha_{j+1}$ (lower edge) or $\lambda=\beta_{j}$ (upper edge) for some $j \in \mathbb{N}$. 
\end{itemize}

In Case I, the Green's function $G_{\lambda}(x,y)$ is the Schwartz kernel of the resolvent operator $R_{\lambda,L}:=(L-\lambda)^{-1}$, while in Case II, it is the Schwartz kernel of the weak limit of resolvent operators $R_{\lambda,L}:=(L-\lambda \pm \varepsilon)^{-1}$ as $\varepsilon \rightarrow 0$ (the sign $\pm$ depends on whether $\lambda$ is an upper or a lower spectral edge).
Note that in the flat case $X=\mathbb{R}^d$, Green's function asymptotics of periodic elliptic operators were obtained in \cite{KKR} for Case I ($d \geq 2$), and in \cite{KR} for Case II ($d \geq 3$). As in \cite{KKR, KR}, we will deduce all asymptotics from an assumed ``generic" spectral edge behavior of the \textbf{dispersion relation} of the operator $L$, which we will briefly review in Section 2.

The organization of the paper is as follows. In Subsection 2.1, we will review some general notions and results about group actions on abelian coverings. Then in Subsection 2.2, we introduce additive and multiplicative functions defined on an abelian covering, which will be needed for writting down the main formulae of Green's function asymptotics. Subsection 2.3 contains not only a brief introduction to periodic elliptic operators on abelian coverings, but also the necessary notations and assumptions for formulating the asymptotics. The main results of this paper are stated in Section 3. In Section 4, the Floquet-Bloch theory is recalled and the problem is reduced to studying a scalar integral. Some auxiliary statements that appeared in \cite{KKR, KR} are collected in Section 5, and the final proofs of the main results are provided in Section 6. Section 7 provides the proofs of some technical claims that were postponed from previous sections. Section 8 discusses analogous results for Green's functions of nonsymmetric periodic elliptic operators of second-order on abelian coverings below and at the generalized principal eigenvalues, and then describes the corresponding Martin compactifications and the Martin integral representations for such operators.
%These generalize the results in the Euclidean case from \cite{MT}.
The last sections contain some concluding remarks and acknowledgements.

%%%%%%%%%%%%%%%%%%%%%%%%%%%
\section{Notions and preliminary results}
\label{setup}
%%%%%%%%%%%%%%%%%%%%%%%%%%%
\subsection{Group actions and abelian coverings}

Let $X$ be a noncompact smooth Riemannian manifold of dimension $n$ equipped with an \textbf{isometric, properly discontinuous, free}, and \textbf{co-compact} action of an \textbf{finitely generated abelian discrete group} $G$. The action of an element $g \in G$ on $x \in X$ is denoted by $g\cdot x$. Due to our conditions, the orbit space $M=X/G$ is a \textbf{compact} smooth Riemannian manifold of dimension $n$ when equipped with the metric pushed down from $X$.
We assume that $X$ and $M$ are connected.
Thus, we are dealing with a normal abelian covering of a compact manifold
$$X \xrightarrow{\pi} M(=X/G),$$
where $G$ is the deck group of the covering $\pi$.

Let $d_X(\cdot, \cdot)$ be the distance metric on the Riemannian manifold $X$. It is known that $X$ is a complete Riemannian manifold since it is a Riemannian covering of a compact Riemannian manifold $M$ (see e.g., \cite{Cha}). Thus, for any two points $p$ and $q$ in $X$, $d_X(p,q)$ is the length of a length minimizing geodesic connecting these two points. 

Let $S$ be any finite generating set of the deck group $G$. We define the \textbf{word length $|g|_S$} of $g \in G$ to be the number of generators in the shortest word representing $g$ as a product of elements in $S$:
$$|g|_S=\min\{n \in \mathbb{N} \mid g=s_1 \dots s_n, s_i \in S\cup S^{-1}\}.$$
The \textbf{word metric $d_S$ on G with respect to S} is the metric on $G$ defined by the formula
$$d_S(g,h)=|g^{-1}h|_S$$
for any $g, h \in G$.

We introduce a notion in geometric group theory due to Gromov that we will need here (see e.g., \cite{Luck, NY}).
\begin{defi}
Let $Y, Z$ be metric spaces. A map $f: Y \rightarrow Z$ is called a \textbf{quasi-isometry}, if the following conditions are satisfied:
\begin{itemize}
\item
There are constants $C_1, C_2>0$ such that
$$C_1^{-1}d_Y(x,y)-C_2 \leq d_{Z}(f(x),f(y)) \leq C_1 d_Y(x,y)+C_2$$
for all $x,y \in Y$.

\item The image $f(Y)$ is a net in $Z$, i.e., there is some constant $C>0$ so that if $z\in Z$, then there exists $y \in Y$ such that $d_Z(f(y),z)<C$.
\end{itemize}
\end{defi}
We remark that given any two finite generating sets $S_1$ and $S_2$ of $G$, the two word metrics $d_{S_1}$ and $d_{S_2}$ on $G$ are equivalent (see e.g., \cite[Theorem 1.3.12]{NY}).

The next result, which directly follows from the \v{S}varc-Milnor lemma (see e.g.,  \cite[Lemma 2.8]{Luck}, \cite[Proposition 1.3.13]{NY}), establishes a quasi-isometry between the word metric $d_S(\cdot, \cdot)$ of the deck group $G$ and the distance metric $d_X(\cdot, \cdot)$ of the Riemannian covering $X$ of a closed connected Riemannian manifold $M$. 

\begin{prop}
\label{P:quasi_iso}
For any $x \in X$, the map
\begin{equation*}
\begin{split}
(G, d_S) &\rightarrow (X, d_X) \\
g &\mapsto g \cdot x
\end{split}
\end{equation*}
given by the action of the deck transformation group $G$ on $X$ is a quasi-isometry.
\end{prop}

Since $G$ is a finitely generated abelian group, its torsion free subgroup is a free abelian subgroup $\mathbb{Z}^d$ of finite index. Hence, we obtain a normal $\mathbb{Z}^d$-covering
$$X \rightarrow M'(=X/\mathbb{Z}^d),$$
and a normal covering of $M$ with a finite number of sheets
$$M' \rightarrow M.$$
Then $M'$ is still a compact Riemannian manifold.
By switching to the normal subcovering $X \xrightarrow{\mathbb{Z}^d} M'$, we \textbf{assume from now on that the deck group $G$ is $\mathbb{Z}^d$} and substitute $M'$ for $M$. This will not reduce generality of our results \footnote{The same reduction holds for any \textbf{finitely generated virtually abelian} deck group $G$.}.

\begin{notation}
\begin{enumerate}[(a)]
\item
Hereafter, we choose the symmetric set $\{-1,1\}^d$ to be the generating set $S$ of $\bZ^d$. Then the function $z=(z_1, \dots, z_d) \mapsto \sum_{j=1}^d |z_j|$ is the word length function $|\cdot|_S$ on $\bZ^d$ associated with $S$. 

\item For a general Riemannian manifold $Y$, we denote by $\mu_Y$ the Riemannian measure of $Y$. We use the notation $L^2(Y)$ for the Lebesgue function space $L^2(Y, \mu_Y)$. Also, the notation $L^2_{comp}(Y)$ stands for the subspace of $L^2(Y)$ consisting of compactly supported functions.
It is worth mentioning that in our case, the Riemannian measure $\mu_X$ is the lifting of the Riemannian measure $\mu_M$ to $X$. Thus, $\mu_X$ is a $G$-invariant Riemannian measure on $X$.

\item 
We recall that a fundamental domain $F(M)$ for $M$ in $X$ (with respect to the action of $G$) is an open subset of $X$ such that 
for any $g \neq e$, $\displaystyle F(M) \cap g\cdot F(M)=\emptyset$ and the subset
$$X\setminus\bigcup_{g \in G} g\cdot F(M)$$
has measure zero. One can refer to \cite{Atiyah} for constructions of such fundamental domains. Henceforth, we use the notation $F(M)$ to stand for a fixed fundamental domain for $M$ in $X$.
\end{enumerate}
\end{notation}

\begin{remark}
%Since the subset $\displaystyle \bigcup_{g \in G} g\cdot F(M)$ is dense in $X$
%and $\bigcup_{g \in G} g\cdot \overline{F(M)}$ is closed in $X$, it follows that 
The closure of $F(M)$ contains at least one point in $X$ from every orbit of $G$, i.e.,
\begin{equation}
\label{closure_F(M)}
X=\bigcup_{g \in G} g\cdot \overline{F(M)}.
\end{equation}
Thus, if $F: X \rightarrow \mathbb{R}$ is the lifting of an integrable function $f: M \rightarrow \mathbb{R}$ to $X$, then 
\begin{equation}
\label{integralF(M)}
\int_{M} f(x)d\mu_M(x)=\int_{\overline{F(M)}} F(x)d\mu_X(x).
\end{equation}
\end{remark}
%%%%%%%%%%%%%%%%%%%%%%%%%%%%%%%%%%%%
%%%%%%%%%%%%%%%%%%%%%%%%%%%%%%%%%%%%
\subsection{Additive and multiplicative functions on abelian coverings}

To formulate our main results in Section 3, we need to introduce an analog of exponential type functions on the noncompact covering $X$. 

We begin with a notion of additive and multiplicative functions on $X$ (see \cite{LinPinchover}). 
%\footnote{The definition can apply to any covering manifold with a discrete deck group.}
\begin{defi}
\label{add_mul}
\begin{itemize}
\item
A real smooth function $u$ on $X$ is said to be \textbf{additive} if there is a homomorphism $\alpha: G \rightarrow \mathbb{R}$ such that 
\begin{equation*}
u(g\cdot x)=u(x)+\alpha(g), \quad \mbox{for all} \quad (g,x) \in G \times X.
\end{equation*}

\item
A real smooth function $v$ on $X$ is said to be \textbf{multiplicative} if there is a homomorphism $\beta$ from $G$ to the multiplicative group $\mathbb{R}\setminus \{0\}$ such that 
\begin{equation*}
v(g\cdot x)=\beta(g) v(x), \quad \mbox{for all} \quad (g,x) \in G \times X.
\end{equation*}

\item
Let $m \in \mathbb{N}$. A function $h$ (resp. $H$) that maps $X$ to $\mathbb{R}^m$ is called a vector-valued additive (resp. multiplicative) function on $X$ if every component of $h$ (resp. $H$) is also additive (resp. multiplicative) on $X$.
\end{itemize}
\end{defi}

Following \cite{KP2, LinPinchover}, we can define explicitly some additive and multiplicative functions for which the group homomorphisms $\alpha$, $\beta$ appearing in Definition \ref{add_mul} are trivial.

\begin{defi}
\label{Exponential}
Let $f$ be a nonnegative function in $C_{c}^{\infty}(X)$ such that $f$ is strictly positive on $\overline{F(M)}$. For any $j=1, \dots, d$, we define the following function
\begin{equation*}
\label{exp_function}
H_j(x)=\sum_{g \in \bZ^d}\exp{(-g_j)}f(g\cdot x).
\end{equation*}
We also put $H(x):=(H_1(x),\dots, H_d(x))$.
\end{defi}
Then $H_j$ is a positive function satisfying the multiplicative property $H_j(g\cdot x)=\exp{(g_j)}H_j(x)$, for any $g=(g_1, \dots, g_d) \in \bZ^d$. The multiplicative function $H$ plays a similar role to the one played by the exponential function $e^{x}$ on the Euclidean space $\mathbb{R}^d$.

By taking logarithms, we obtain an additive function on $X$, which leads to the next definition.
\begin{defi}
\label{Add_func}
We introduce the following smooth $\mathbb{R}^d$-valued function on $X$: 
$$h(x):=(\log{H_1(x)},  \cdots, \log{H_d(x)}).$$
Then $h=(h_1, \dots, h_d)$ with $h_j(x)=\log{H_j(x)}$. Thus, $h$ satisfies the following additivity:
\begin{equation}
\label{additivity}
h(g \cdot x)=h(x)+g, \quad \mbox{for all} \quad (g,x) \in G \times X.
\end{equation}
Here we use the natural embedding $G=\bZ^d \subset \mathbb{R}^d$.
\end{defi}
Clearly, the definitions of functions $H$ and $h$ depend on the choice of the function $f$ and the fundamental domain $F(M)$. So, there is no canonical choice for constructing additive and multiplicative functions. Nevertheless, a more invariant approach to defining additive and multiplicative functions on Riemannian co-compact coverings can be found in \cite{Ag2, KP2, Kha}.

The following important comparison between the Riemannian metric and the distance from the additive function $h$ in Definition \ref{Add_func} will be needed later.
\begin{prop}
\label{P:add_func}
There are some positive constants $R_h$ (depending on $h$) and $C>1$ such that
whenever $d_X(x,y)\geq R_h$, we have 
$$C^{-1}\cdot d_X(x,y) \leq |h(x)-h(y)| \leq C\cdot d_X(x,y).$$
Here $|\cdot|$ is the Euclidean distance on $\mathbb{R}^d$, and the constant $C$ is independent of the choice of $h$.

As a consequence, the pseudo-distance $d_h(x,y):=|h(x)-h(y)| \rightarrow \infty$ if and only if $d_X(x,y) \rightarrow \infty$.
\end{prop}
The proof of this statement is given in Section 7.

\begin{defi}
\label{admissible_set}
For any additive function $h$ satisfying \mref{additivity}, $\mathcal{A}_h$ is the set consisting of unit vectors $s \in \mathbb{S}^{d-1}$ such that there exist two points $x$ and $y$ satisfying $d_X(x,y)>R_h$ and
$$s=(h(x)-h(y))/|h(x)-h(y)|.$$
The set $\mathcal{A}_h$ is called the \textbf{admissible set of the additive function $h$}, and its elements are \textbf{admissible directions} of $h$.
\end{defi}
For the proof of the following proposition, one can see in Section 7.
\begin{prop}
\label{A_h}
For any additive function $h$ on $X$, one has 
\begin{equation}
\label{rational}
\mathbb{Q}^{d}\cap \mathbb{S}^{d-1}=\{ g/|g| \mid g \in \bZ^d\setminus \{0\}\} \subset \mathcal{A}_h.
\end{equation}
Hence, the admissible set $\mathcal{A}_h$ of $h$ is dense in the sphere $\mathbb{S}^{d-1}$.
In particular, when $d=2$, $\mathcal{A}_h$ is the whole unit circle $\mathbb{S}^1$.
\end{prop}

\begin{remark}
When the dimension $n$ of $X$ is less than $(d-1)/2$ (e.g., if $d>5$ and $X$ is the standard two dimensional jungle gym $JG^2$ in $\mathbb{R}^d$, see \cite{Pin}), the $(d-1)$-dimensional Lebesgue measure of the admissible set $\mathcal{A}_h$ of any additive function $h$ on $X$ is zero. To see this, we first denote by $X_h$ the $2n$-dimensional smooth manifold $\{(x,y) \in X \times X \mid d_X(x,y)>R_h\}$, and then consider the smooth mapping:
\begin{equation*}
\begin{split}
\Psi:  X_h &\rightarrow \mathbb{S}^{d-1}
\\       (x,y) &\mapsto \frac{h(x)-h(y)}{|h(x)-h(y)|}.
\end{split}
\end{equation*}
Then $\mathcal{A}_h$ is the range of $\Psi$. Since $\dim{X_h}<\dim{\mathbb{S}^{d-1}}$, every point in the range of $\Psi$ is  critical and thus, $\mathcal{A}_h$ has measure zero by Sard's theorem.
\end{remark}
\begin{example}
\label{ex_add}
\begin{itemize}
\item
Here is a family of non-trivial examples of additive functions in the flat case, i.e., when the covering space $X$ is $\mathbb{R}^d$ and the base is the $d$-dimensional torus $\mathbb{T}^d$.
Let $d \geq 1$ and $\varphi$ be a real smooth function in $\mathbb{R}^d$ such that $\varphi$ is $\bZ^d$-periodic.
It is shown in \cite{Ancona} that there exists a unique map $F_{\varphi}=((F_{\varphi})_1, \dots, (F_{\varphi})_d): \mathbb{R}^d \rightarrow \mathbb{R}^d$ satisfying $F_{\varphi}(0)=0$, the additive condition \mref{additivity}, i.e.,
$F_{\varphi}(x+n)=F_{\varphi}(x)+n$ for any $(x,n) \in \mathbb{R}^d \times \bZ^d$, and the equation
$$\Delta (F_{\varphi})_i=\nabla \varphi \cdot \nabla (F_{\varphi})_i,$$
for any $1 \leq i \leq d$. Note that $F_{\varphi}$ is just the identity mapping in the trivial case when $\varphi=0$. Moreover, it is also known \cite{Ancona} that when $d=2$, $F_{\varphi}$ is a diffeomorphism of $\mathbb{R}^d$ onto itself. In particular, for any $\bZ^2$-periodic function $\varphi$, 
$|F_{\varphi}(x)-F_{\varphi}(y)| \geq C_{\varphi}|x-y|$ for any $x,y \in \mathbb{R}^2$ for some $C_{\varphi}>0$.
However, when $d\geq 3$, $F_{\varphi}$ may admit a critical point for some $\bZ^d$-periodic function $\varphi$.

\item Let $X \xrightarrow{p} M$, $Y \xrightarrow{q} N$ be normal $\bZ^{d_1}$ and $\bZ^{d_2}$ coverings of compact Riemannian manifolds $M$ and $N$ respectively. Then $X \times Y \xrightarrow{p \times q} M \times N$ is also a normal $\bZ^{d_1+d_2}$ covering of $M \times N$. Consider any $\mathbb{R}^{d_1}$-valued function $h_1$ (resp. $\mathbb{R}^{d_2}$-valued function $h_2$) defined on $X$ (resp. $Y$). Let us denote by $h_1 \oplus h_2$ the following $\mathbb{R}^{d_1+d_2}$-valued function on $X \times Y$: 
$$(h_1 \oplus h_2) (x,y)=(h_1(x), h_2(y)), \quad (x,y) \in X \times Y.$$
Then it is clear that $h_1 \oplus h_2$ is additive (resp. multiplicative) on $X \times Y$ if and only if both functions $h_1$ and $h_2$ are additive (resp. multiplicative). Moreover, $\mathcal{A}_{h_1 \oplus h_2} \subseteq \{\left(a_1\cdot\mathcal{A}_{h_1}, a_2 \cdot \mathcal{A}_{h_2}\right) \mid 0<a_1, a_2<1 \hspace{4pt} \mbox{and} \hspace{4pt} a_1^2+a_2^2=1\}$.
\end{itemize}
\end{example}
\subsection{Some notions and assumptions}
\label{setup2.3}
Let $L$ be a \textbf{bounded from below} and \textbf{symmetric second-order elliptic}\footnote{The ellipticity is understood in the sense of the nonvanishing of the principal symbol of the operator $L$ on the cotangent bundle of the underlying manifold (with the zero section removed).} operator on $X$ with \textbf{smooth}\footnote{The smoothness condition is assumed for avoiding lengthy technicalities and it can be relaxed.} coefficients such that the operator \textbf{commutes with the action of $G$}.
An operator that commutes with the action of $G$ is called a \textbf{$G$-periodic} (or sometimes \textbf{periodic}) operator for brevity.

Notice that on a Riemannian co-compact covering, any $G$-periodic elliptic operator with smooth coefficients is \textbf{uniformly elliptic} in the sense that
\begin{equation*}
\label{uniform_ell}
|L_0^{-1}(x,\xi)|\leq C|\xi|^{-2}, \quad (x,\xi) \in T^*X, \xi \neq 0.
\end{equation*}
Here $|\xi|$ is the Riemannian length of $(x,\xi)$ and $L_0(x,\xi)$ is the principal symbol of $L$.

The periodic operator $L$ can be pushed down to an elliptic operator $L_M$ on $M$ and thus, $L$ is the lifting of an elliptic operator $L_M$ to $X$. By a slight abuse of notation, we will use the same notation $L$ for both elliptic operators acting on $X$ and $M$. 

Under these assumptions on $L$, the symmetric operator $L$ with the domain $C^{\infty}_{c}(X)$ is \textbf{essentially self-adjoint} in $L^2(X)$, i.e., the minimal operator $L_{min}$ coincides with the maximal operator $L_{max}$ (see e.g., \cite{Shubin_spectral} for notation $L_{min}$ and $L_{max}$). 
This fact can be found in \cite[Proposition 3.1]{Atiyah}, for instance \footnote{In \cite{Atiyah}, Atiyah proved for symmetric elliptic operators acting on Hermitian vector bundles over any general co-compact covering manifold (not necessary to be a Riemannian covering). Later, in \cite{BruSun1}, Brunning and Sunada extended Atiyah's arguments to the case including compact quotient space $X/G$ with singularities.
}. 
Hence, there exists a unique self-adjoint extension in the Hilbert space $L^2(X)$ of $L$, which we denote also by $L$. Since $L$ is a uniformly elliptic operator on the manifold $X$ of bounded geometry, its domain is the Sobolev space $H^2(X)$ \cite[Proposition 4.1]{Shubin_spectral}, and henceforward, we always work with this self-adjoint operator $L$.

\begin{notation}
\begin{enumerate}[(a)]
\item
The \textbf{dual} (or \textbf{reciprocal}) \textbf{lattice} is $2\pi \mathbb{Z}^d$ and its fundamental domain is the cube $[-\pi,\pi]^{d}$ (\textbf{Brillouin zone}).

\item For any $m \in \mathbb{N}$, the $m$-dimensional torus $\mathbb{R}^m/\mathbb{Z}^m$,  is denoted by $\mathbb{T}^m$.
\end{enumerate}
\end{notation}

From now on, we fix any smooth function $h$ satisfying \mref{additivity} in Definition \ref{Add_func}. The following lemma is a preparation for the next definition. 

\begin{lemma}
\label{fiber_op_structure}
For any $k \in \mathbb{C}^d$, we have
$$e^{-ik\cdot h(x)}L(x,D)e^{ik\cdot h(x)}=L(x,D)+B(k),$$
where $B(k)$ is a smooth differential operator of order $1$ on $X$ that commutes with the action of the deck group $G$.
Thus by pushing down, the differential operators $e^{-ik\cdot h(x)}L(x,D)e^{ik\cdot h(x)}$ and $B(k)$ can be considered also as differential operators on $M$.
Moreover, given any $m \in \mathbb{R}$, the mapping
$$k \mapsto e^{-ik\cdot h(x)}L(x,D)e^{ik\cdot h(x)}$$
is analytic in $k$ as a $B(H^{m+2}(M), H^{m}(M))$-valued function.
\end{lemma}

\begin{proof}
It is standard that the commutator $[L, e^{ik\cdot h(x)}]$ is a differential operator of order $1$ on $X$.
Now one can write 
\begin{equation*}
\label{B(k)}
B(k)=e^{-ik\cdot h(x)}Le^{ik\cdot h(x)}-L=e^{-ik\cdot h(x)}[L, e^{ik\cdot h(x)}]
\end{equation*}
to see that $B(k)$ is also a smooth differential operator of order $1$. Also, one can check that $B(k)$ commutes with the action of $G$ by using $G$-periodicity of the operator $L$ and additivity of $h$. This proves the first claim of the lemma. From a standard fact (see e.g., \cite[Theorem 2.2]{GUI}), the operator $e^{-ik\cdot h(x)}Le^{ik\cdot h(x)}$ defined on $X$ can be written as a sum $\sum_{|\alpha|\leq 2}k^{\alpha}L_{\alpha},$
where $L_{\alpha}$ is a $G$-periodic differential operator on $X$ of order $2-|\alpha|$ which is independent of $k$. By pushing the above sum down to a sum of operators on $M$, the claim about analyticity in $k$ is then obvious.
\end{proof}

\begin{defi}
\label{fiber_op}
For any $k \in \mathbb{C}^d$, we denote by $L(k)$ the elliptic operator 
\begin{equation*}
\label{conjugatingLk}
e^{-ik\cdot h(x)}L(x,D)e^{ik\cdot h(x)}
\end{equation*}
in $L^2(M)$ with the domain the Sobolev space $H^2(M)$.

In this definition, the vector $k$ is called the \textbf{quasimomentum} \footnote{The name comes from solid state physics \cite{AshMer}.}. 
\end{defi}

\begin{remark}
\label{remark_L(k)}
\begin{enumerate}[(a)]
\item
When dealing with real quasimomentum $k$, it is enough to consider $k$ in any shifted copy of the Brillouin zone $[-\pi, \pi]^d$, since the operators $L(k)$ and $L(k+2\pi \gamma)$ are unitarily equivalent, for any $\gamma \in \mathbb{Z}^d$. 

\item
The operator $L(k)$ is self-adjoint in $L^{2}(M)$ for each $k \in \mathbb{R}^d$, with the domain $H^{2}(M)$. Due to the ellipticity of $L$, each of the operators $L(k)$ ($k \in \mathbb{R}^d$) has discrete real spectrum and thus, we can list its eigenvalues in non-decreasing order:
\begin{equation*}
\label{eigenv}
\lambda_{1}(k) \leq \lambda_{2}(k) \leq ... \quad .
\end{equation*}
Hence, we can single out continuous and piecewise-analytic \textbf{band functions} $\lambda_{j}(k)$ for each $j \in \mathbb{N}$ \cite{Wilcox}.

\item
By Lemma \ref{fiber_op_structure}, the operators $L(k)$ are perturbations of the self-adjoint operator $L(0)$ by lower order operators $B(k)$ for each $k \in \mathbb{C}^d$. Consequently, the spectra of the operators $L(k)$ on $M$ are all discrete (see \cite[pp.180-190]{Agmon}).

\item
We now describe another equivalent model of the operators $L(k)$, which sometimes can be useful (see \cite{KP2}). For any quasimomentum $k \in \mathbb{C}^d$, we denote by $\gamma_k$ the character (i.e., a $1$-dimensional representation) $e^{ik\cdot g}$ of the abelian group $G$ and consider the  $1$-dimensional flat vector bundle $E_k$ over $M$ associated with this representation. For any real number $s$, let $H^{s}_k(X)$ be the space of $H^s$-sections of $E_k$. Since $L$ is $G$-periodic, $L$ maps continuously $H^2_k(X)$ into $L^2_k(X)$. This defines an elliptic operator over the space $\mathcal{E}(M, E_k)$ of smooth sections of $E_k$ over the compact manifold $M$. Moreover, this elliptic operator is unitarily equivalent to the operator $L(k)$ in Definition \ref{fiber_op}. 
Notice that in the case of abelian coverings, a third equivalent model using differential forms and the Jacobian torus $J(M)$ to define can be found in \cite{Sunada1}.
\end{enumerate}
\end{remark}

Now we can restate the band-gap structure of $\sigma(L)$ presented in Theorem \ref{bandgap} in more details.
\begin{thm}
 \cite{BruSun1, Ea, KOS,K, RS4,Sunada1,Sunada2}
The spectrum of $L$ is the union of all the spectra of $L(k)$ when $k$ runs over the Brillouin zone (or any its shifted copy), i.e.
\begin{equation}
\label{fl_spectrum}
\sigma(L)=\bigcup_{k \in [-\pi, \pi]^d}\sigma(L(k)).
\end{equation}
In other words, the spectrum of $L$ is the range of the multivalued function
\begin{equation*} 
\label{sp_function}
k \rightarrow \lambda(k):=\sigma(L(k)), \quad  k\in  [-\pi, \pi]^d,
\end{equation*}
 As a result, the range of the band function $\lambda_j$ (see remark \ref{remark_L(k)}) constitutes exactly the band $[\alpha_j, \beta_j]$ of the spectrum of $L$  shown in \mref{band-gap}.
\end{thm}

The notions which we will introduce now are important concepts in solid state physics (see e.g., \cite{AshMer}) as well as in general theory of periodic elliptic operators (see e.g., \cite{K}).
\begin{defi}
\begin{itemize}
\item
A \textbf{Bloch solution with quasimomentum $k$} of the equation $L(x,D)u=0$ is a solution of the form
\begin{equation*}
u(x)=e^{ik\cdot h(x)}\phi (x),
\end{equation*}
where $h$ is any fixed additive function on $X$ and the function $\phi$ is invariant under the action of the deck transformation group $G$.\footnote{It is easy to see that this definition is independent of the choice of $h$.}
\item
The \textbf{Bloch variety} $B_{L}$ of the operator $L$ consists of all pairs $(k,\lambda) \in \mathbb{C}^{d+1}$ such that the equation $Lu=\lambda u$ on $X$ has a non-zero Bloch solution $u$ with quasimomentum $k$. 
The Bloch variety $B_{L}$ can be seen as the graph of the multivalued function $\lambda(k)$, which is also called the \textbf{dispersion relation}:
\begin{equation*}
\label{disp_relation}
B_L=\{(k,\lambda): \lambda \in \sigma(L(k))\}.
\end{equation*}

\item
The \textbf{Fermi surface} $F_{L,\lambda}$ of the operator $L$ at the energy level $\lambda \in \mathbb{C}$ consists of all quasimomenta $k \in \mathbb{C}^{d}$ such that the equation $Lu=\lambda u$ on $X$ has a non-zero Bloch solution $u$ with quasimomentum $k$.  We shall write $F_{L}$ instead of $F_{L,0}$ when $\lambda=0$.
Equivalently, Fermi surfaces are level sets of the dispersion relation.
\end{itemize}
\end{defi}

The next statement can be found in \cite[Theorem 3.1.7]{K} (see also \cite[Lemma 8]{KP2}).
\begin{lemma}
\label{L:Bloch_variety}
There exist entire ($2\pi \mathbb{Z}^d$-periodic in $k$) functions of finite orders on $\mathbb{C}^{d}$ and on $\mathbb{C}^{d+1}$ such that the Fermi and Bloch varieties are the sets of all zeros of these functions respectively.
As a consequence, the band functions $\lambda_j(k)$ are piecewise analytic on $\mathbb{C}^d$. 
\end{lemma}
Note that the piecewise analyticity of the band functions is shown initially in \cite{Wilcox} for Schr\"odinger operators in the flat case.

Without loss of generality, it is enough to assume henceforth that $0$ is the spectral edge of interest (by adding a constant into the operator $L$ if neccessary) and there is a spectral gap below this spectral edge $0$. Therefore,
$0$ is the lower spectral edge of some spectral band \footnote{The upper spectral edge case is treated similarly.}, i.e., $0$ is the minimal value of some band function $\lambda_j(k)$ for some $j \in \mathbb{N}$  over the Brillouin zone. 

As in \cite{KKR, KR}, the following analytic assumptions are imposed on the band function $\lambda_j$:\\

\textbf{Assumption A}\\

\emph{There exists $k_0 \in [-\pi, \pi]^d$ and a band function $\lambda_{j}(k)$ such that:}\\

\textbf{A1} $\lambda_{j}(k_0)=0$.\\

\textbf{A2} $\min_{k \in \mathbb{R}^d, i \neq j}|\lambda_{i}(k)|>0$.\\

\textbf{A3} \emph{$k_0$ is the only (modulo $2\pi \mathbb{Z}^d$) minimum of $\lambda_{j}$}.\\

\textbf{A4} \emph{The Hessian matrix $H:=\Hess{(\lambda_j)}(k_0)$ of $\lambda_j$ at $k_0$  is positive-definite}.\\

\textbf{A5} \emph{All components of the quasimomentum $k_0$ are equal to either $0$ or $\pi$.
}\\

\begin{remark}
\begin{enumerate}[(a)]
\item
For the flat case, the main theorem in \cite{KloppRalston} shows that the conditions \textbf{A1} and \textbf{A2} are `generically' satisfied, i.e., they can be achieved by small perturbation of the potential of a periodic Schr\"odinger operator. The same proof in \cite{KloppRalston} still works for periodic Schr\"odinger operators on a general abelian covering. 

\item
In mathematics and physics literature, the conditions \textbf{A3} and \textbf{A4} are commonly believed to be `generically' true (see e.g., \cite[Conjecture 5.1]{KP2}). In particular, \textbf{A4} is often assumed to define \textbf{effective masses} of Bloch electrons \cite{AshMer}. Additionally, we remark that the condition \textbf{A3} can be relaxed (see Section 9).

\item
It is known \cite{HKSW} that spectral edges could occur deeply inside the Brillouin zone, however, the condition \textbf{A5} holds in many practical cases. We shall only use this condition for the \textit{spectral gap interior} case.

\item
Due to results of \cite{KS} (in the flat case) and of \cite{KOS} (in the general case), all these assumptions \textbf{A1-A5} hold at the bottom of the spectrum for non-magnetic Schr\"odinger operators.
\end{enumerate}
\end{remark}

Here are some notations that will be used thoughout this paper.
\begin{notation}
\label{notation}
\begin{enumerate}[(a)]
\item The real parts of a complex vector $z$ and of a complex matrix $A$ are denoted by $\Re(z)$ and $\Re(A)$, respectively.  

\item For any two functions $f$ and $g$ defined on $X \times X$, if there exist constants $C>0$ and $R>0$ such that $|f(x,y)|\leq C|g(x,y)|$ whenever $d_X(x,y)>R$, we write $f(x,y)=O(g(x,y))$. 
\end{enumerate}
\end{notation}

We say that a set $W$ in $\mathbb{C}^d$ is \textbf{symmetric} 
if for any $z \in W$, we have $\overline{z} \in W$.

The following proposition will play a crucial role in establishing Theorem \ref{main}. 
\begin{prop}
\label{analytic_perturbation}
There exists an $\epsilon_0>0$ and a symmetric open subset $V \subset \mathbb{C}^d$ containing the quasimomentum $k_0$ from Assumption A such that the band function $\lambda_j$ in Assumption A has an analytic continuation into a neighborhood of $\overline{V}$, and the following properties hold for any $z$ in a symmetric neighborhood of $\overline{V}$:
\begin{enumerate}[\bf{(P}1)]
\item $\lambda_{j}(z)$ is a simple eigenvalue of $L(z)$.

\item $|\lambda_j(z)|<\epsilon_0$ and $\displaystyle \overline{B}(0,\epsilon_0) \cap \sigma(L(z))=\{\lambda_j(z)\}$.

\item  There is a nonzero $G$-periodic function $\phi_z$ defined on $X$ such that $$L(z)\phi_z=\lambda_j(z)\phi_z.$$
Moreover, $z \mapsto \phi_z$ can be chosen analytic as a $H^2(M)$-valued function.

\item $\displaystyle 2\Re(\Hess{(\lambda_{j})}(z))>\min \sigma(\Hess{(\lambda_{j})}(k_0))\cdot I_{d \times d}$. 

\item
$\displaystyle F(z):=(\phi_{z}(\cdot),\phi_{\overline{z}}(\cdot))_{L^{2}(M)} \neq 0.$
\end{enumerate}
\end{prop}

\begin{proof}
Due to Remark \ref{remark_L(k)}, for any $z \in \mathbb{C}^d$, the operator $L(z)$ has discrete spectrum and thus, it is a closed operator with nonempty resolvent set. Moreover, the operator domain $H^2(M)$ of $L(z)$ is independent of $z$. Also, by Lemma \ref{fiber_op_structure}, for any $\phi \in H^2(M)$, $L(z)\phi$ is a $L^2(M)$-valued analytic function of $z$. These imply that $\{L(z)\}_{z \in \mathbb{C}^d}$ is an analytic family of type $\mathcal{A}$ (see e.g., \cite{Ka, RS4}).
Now \textbf{(P1)-(P4)} would follow easily from analytic perturbation theory \cite{RS4} using conditions \textbf{A1}, \textbf{A2} and \textbf{A4}, while \textbf{(P5)} is due to \textbf{(P3)} and the inequality $F(k_0)=\|\phi_{k_{0}}\|^{2}_{L^{2}(M)}>0$.
\end{proof}
Define $\mathcal{V}:=\{\beta \in \mathbb{R}^d \mid k_0+i\beta \in \overline{V}\}.$
Now we introduce the function 
$E(\beta):=\lambda_{j}(k_0+i\beta)$, which is defined on $\mathcal{V}$.

The next lemma (see \cite{KKR}) is the only place in this paper where the condition \textbf{A5} is used.
\begin{lemma}
\label{func_E}
Assume that the operator $L$ is \textbf{real} \footnote{Namely, $Lu$ is real whenever $u$ is real.} and the condition \textbf{A5} is satisfied. Then $E$ is a \textbf{real-valued} function. By reducing the neighborhood $V$ in Proposition \ref{analytic_perturbation} if necessary, the function $E$ can be assumed real analytic and strictly concave function from $\mathcal{V}$ to $\mathbb{R}$ such that its Hessian at any point $\beta$ in $\mathcal{V}$ is negative-definite.
\end{lemma}

For $\lambda \in \mathbb{R}$, we put

$$K_{\lambda}:=\{\beta \in \mathcal{V}: E(\beta)\geq \lambda \}$$
and
$$\Gamma_{\lambda}:=\{\beta \in \mathcal{V}: E(\beta)=\lambda \}$$

Due to Lemma \ref{func_E}, $K_{\lambda}$ is a strictly convex $d$-dimensional compact set in $\mathbb{R}^d$, and its boundary $\Gamma_{\lambda}$ is a compact hypersurface in $\mathbb{R}^d$ whose Gauss-Kronecker curvature is nowhere zero. 
Therefore, there exists a diffeomorphism $\beta$ from $\mathbb{S}^{d-1}$ onto $\Gamma_{\lambda}$ such that 
\begin{equation*}
\label{E:gradient_E_s}
\nabla E(\beta_{s})=-|\nabla E(\beta_s)|s.
\end{equation*}
In addition, $$\lim_{|\lambda| \rightarrow 0}\max_{s \in \mathbb{S}^{d-1}}|\beta_s|=0.$$

By letting $|\lambda|$ be sufficiently small, 
we will suppose that there is an $r_0>0$ (independent of $s$) such that 
\begin{equation}
\label{E:beta_s_in_V}
\{k+it\beta_{s} \mid (t,s) \in [0,1] \times \mathbb{S}^{d-1}, \hspace{3pt} |k-k_0|\leq r_0\} \subset V.
\end{equation}

%%%%%%%%%%%%%%%%%%%%%%%%%%%%%%%%%%%%%%
%%%%%%%%%%%%%%%%%%%%%%%%%%%%%%%%%%%%%%
\section{The main results}
\label{mainresults}
We recall that $h$ is a fixed additive function satisfying \mref{additivity} in Definition \ref{Add_func}.

First, we consider the case when $\lambda$ is inside a gap and is near to one of the edges of the gap.
The following result is an analog for abelian coverings of compact Riemannian manifolds of \cite[Theorem 2.11]{KKR}.
\begin{thm}
\label{main} 
(\textit{Spectral gap interior})

Suppose that $d \geq 2$, $L$ is real, and the conditions \textbf{A1-A5} are satisfied.
For $\lambda<0$ sufficiently close to $0$ (depending on the dispersion branch $\lambda_j$ and the operator $L$), the Green's function $G_{\lambda}$ of $L$ at $\lambda$ admits the following asymptotics as $d_{X}(x,y) \rightarrow \infty$:
\begin{equation}
\label{main_asymp}
\begin{split}
G_{\lambda}(x,y)&=\frac{e^{(h(x)-h(y))(ik_{0}-\beta_{s})}}{(2\pi|h(x)-h(y)|)^{(d-1)/2}}\cdot\frac{|\nabla E(\beta_s)|^{(d-3)/2}}{\det{(-\mathcal{P}_s \Hess{(E)}(\beta_{s})\mathcal{P}_s)}^{1/2}}\\& \times \frac{\phi_{k_{0}+i\beta_{s}}(x)\overline{\phi_{k_{0}-i\beta_{s}}(y)}}{(\phi_{k_{0}+i\beta_{s}},\phi_{k_{0}-i\beta_{s}})_{L^{2}(M)}}
+e^{(h(y)-h(x))\cdot \beta_{s}}r(x,y).
\end{split}
\end{equation}
Here $$\displaystyle s=(h(x)-h(y))/|h(x)-h(y)| \in \mathcal{A}_h,$$ 
and $\mathcal{P}_s$ is the projection from $\mathbb{R}^{d}$ onto the tangent space of the unit sphere $\mathbb{S}^{d-1}$ at the point $s$.
Also, there is a constant $C>0$ (independent of $s$ and of the choice of $h$) such that the remainder term $r$ satisfies $$|r(x,y)| \leq Cd_X(x,y)^{-d/2},$$ when $d_X(x,y)$ is large enough.
\end{thm}

By using rational admissible directions (see \mref{rational}) in the formula \mref{main_asymp}, the large scale behaviors of the Green's function along orbits of the $G$-action admit the following nice form in which the additive function $h$ is absent.
%For rational admissible directions  in \mref{rational}, the formula \mref{main_asymp} can be rewritten immediately in a form that reflects how the deck transformation group effects the large scale behaviors of the Green's function along orbits of the $G$-action.
\begin{cor}
Under the same notations and hypotheses of Theorem \ref{main} and suppose that $\lambda<0$ is close enough to $0$, as $|g| \rightarrow \infty$ ($g \in \bZ^d$), we have
\begin{equation}
\label{main_rational}
\begin{split}
G_{\lambda}(x,g \cdot x)&=\frac{e^{g\cdot (ik_{0}-\beta_{g/|g|})}}{(2\pi|g|)^{(d-1)/2}}\cdot\frac{|\nabla E(\beta_{g/|g|})|^{(d-3)/2}}{\det{(-\mathcal{P}_{g/|g|} \Hess{(E)}(\beta_{g/|g|})\mathcal{P}_{g/|g|})}^{1/2}}\\& \times \frac{\phi_{k_{0}+i\beta_{g/|g|}}(x)\overline{\phi_{k_{0}-i\beta_{g/|g|}}(g \cdot x)}}{(\phi_{k_{0}+i\beta_{g/|g|}},\phi_{k_{0}-i\beta_{g/|g|}})_{L^{2}(M)}}
+e^{g\cdot \beta_{s}}O(|g|^{-d/2}).
\end{split}
\end{equation}
\end{cor}

We also give another interpretation of \cite[Theorem 2.11]{KKR} in the special case $X=\mathbb{R}^2$ as follows:
\begin{cor}
Let $\varphi$ be any real, $\bZ^2$-periodic and smooth function on $\mathbb{R}^2$, and we recall the notation $F_{\varphi}$ from Example \ref{ex_add}. Let $s$ be any unit vector in $\mathbb{R}^2$ and $y \in \mathbb{R}^2$. Then as $|t| \rightarrow \infty$ $(t \in \mathbb{R})$, the Green's function $G_{\lambda}$ of $L$ at $\lambda$ ($\approx 0$) has the following asymptotics 
\begin{equation*}
\label{main_flat}
\begin{split}
G_{\lambda}(F_{\varphi}^{-1}(ts+F_{\varphi}(y)),y)&=\frac{e^{ts\cdot(ik_{0}-\beta_{s})}}{(2\pi |\nabla E(\beta_s)|\cdot \det{(-\mathcal{P}_s \Hess{(E)}(\beta_{s})\mathcal{P}_s)\cdot |t|})^{1/2}}\\& \times \frac{\phi_{k_{0}+i\beta_{s}}(F_{\varphi}^{-1}(ts+F_{\varphi}(y)))\overline{\phi_{k_{0}-i\beta_{s}}(y)}}{(\phi_{k_{0}+i\beta_{s}},\phi_{k_{0}-i\beta_{s}})_{L^{2}(\mathbb{T}^2)}}
+e^{ts\cdot \beta_{s}}O(|t|^{-1}).
\end{split}
\end{equation*}
\end{cor}

We now switch to the case when $\lambda$ is on the boundary of the spectrum. Recall that we assume the spectral edge $\lambda$ is zero. The following result is a  generalization of \cite[Theorem 2]{KR}.
\begin{thm}
\label{main_KR} 
(\textit{Spectral edge case})

Assume that $d \geq 3$ and the operator $L$ satisfies the assumptions \textbf{A1-A4}. For a small  $\varepsilon>0$, we denote by $R_{-\varepsilon}=(L+\varepsilon)^{-1}$ the resolvent of $L$ near the spectral edge $\lambda=0$
(which exists, due to Assumption A). Then:
\begin{enumerate}[i)]
\item For any $\phi, \varphi \in L^2_{comp}(X)$, as $\varepsilon \rightarrow 0$, we have:
$$\langle R_{-\varepsilon}\phi, \varphi \rangle \rightarrow \langle R\phi, \varphi \rangle.$$
for an operator $R: L^2_{comp}(X) \rightarrow L^2_{loc}(X)$.
\\

\item The Schwartz kernel $G(x,y)$ of the operator $R$, which we call  \textbf{the Green's function of $L$} (at the spectral edge 0), has the following asymptotics when $d_X(x,y) \rightarrow \infty$:

\begin{equation}
\label{main_asymp_KR}
\begin{split}
G(x,y)&=\frac{\Gamma(\frac{d-2}{2})e^{i(h(x)-h(y))\cdot k_0}}{2\pi^{d/2}\sqrt{\det H}|H^{-1/2}(h(x)-h(y))|^{d-2}}\cdot \frac{\phi_{k_0}(x)\overline{\phi_{k_0}(y)}}{\|\phi_{k_0}\|_{L^2(M)}^2}\\& \times \left(1+O\left(d_X(x,y)^{-1}\right)\right)+O\left(d_{X}(x,y)^{1-d}\right),
\end{split}
\end{equation}
where $H$ is the Hessian matrix of $\lambda_j$ at $k_0$.
\end{enumerate}
\end{thm}
Here the notation $\Gamma(z)$ means the Gamma function $\Gamma(z)=\displaystyle\int_{0}^{\infty}x^{z-1}e^{-x}d{x}$.
\begin{remark}
\begin{enumerate}[(a)]
\item
An interesting feature in the main results is that the dimension $n$ of the covering manifold $X$ does not explicitly enter into the  asymptotics \mref{main_asymp} and \mref{main_asymp_KR} (especially, see also \mref{main_rational}). Anyway, it certainly influences the geometry of the dispersion curves and therefore the asymptotics too. However, as the Riemannian distance between $x$ and $y$ becomes larger, one can see that in the asymptotics, the role of the dimension $n$ is rather limited, while the influence of the rank $d$ of the torsion-free subgroup of the deck group $G$ is stronger.

\item
Note that for a periodic elliptic operator of second order on $\mathbb{R}^d$, at the bottom of its spectrum, the operator is known to be critical when the dimension $d \leq 2$ (see \cite{LinPinchover, MT, Pinsky}). 
This also holds true for the Laplacian on a co-compact Riemannian covering (see \cite[Theorem 5.2.1]{Davies-heat}).
Therefore, the assumption $d \geq 3$ is needed in Theorem \ref{main_KR}.

\item The asymptotics \mref{main_asymp} and \mref{main_asymp_KR} can be described in terms of the Albanese map and the Albanese pseudo-distance on the abelian covering $X$ (see these definitions in \cite[Section 2]{KoSu}), provided that the additive function $h$ is chosen to be harmonic (see also \cite{Kha}).
\end{enumerate}
\end{remark}

Proving Theorem \ref{main_KR} by generalizing \cite[Theorem 2]{KR} is similar to establishing Theorem \ref{main} by generalizing \cite[Theorem 2.11]{KKR}. Thus, after finishing the proof of Theorem \ref{main}, we will sketch briefly the proof of Theorem \ref{main_KR} in Section 6.

We outline the general strategy of both the proofs of Theorem \ref{main} and Theorem \ref{main_KR}.
As in \cite{KKR, KR}, the idea is to show that only one branch of the dispersion relation $\lambda_j$ appearing in the Assumption A will control the asymptotics. 
\begin{itemize}
\item
\textbf{Step 1:} We use the Floquet transform to reduce the problems of finding asymptotics of Green's functions to the problems of obtaining asymptotics of some integral expressions with respect to the quasimomentum $k$. 
\item
\textbf{Step 2:} We localize these expressions around the quasimomentum $k_0$ and then we cut an ``infinite-dimensional" part of the operator to deal only with the multiplication operator by the dispersion branch $\lambda_j$. 
\item
\textbf{Step 3:} The dispersion curve around this part is almost a paraboloid according to the assumption \textbf{A4}, thus,
we can reduce this piece of operator to the normal form in
the free case. In this step, we obtain some scalar integral expressions which are close to the ones arising when dealing with the Green's function of the Laplacian operator at the level $\lambda$. Our remaining task is devoted to computing the asymptotics of these scalar integrals.
\end{itemize}
%%%%%%%%%%%%%%%%%%%%%%%%%%%%%%%%%%%%
%%%%%%%%%%%%%%%%%%%%%%%%%%%%%%%
%%%%%%%%%%%%%%%%%%%%%%%
\section{A Floquet-Bloch reduction of the problem}
\begin{notation}
First, we introduce the following fundamental domain $\mathcal{O}$ (with respect to the dual lattice $2\pi \bZ^d$ and the quasimomentum $k_0$ in Assumption A):
\beq
\label{fund-dom}
\mathcal{O}=k_0+[-\pi, \pi]^d.
\eeq
In another word, $\mathcal{O}$ is just a shifted version of the Brillouin zone so that the quasimomentum $k_0$ is its center of symmetry. 

If $k_0$ is a high symmetry point of the reduced Brillouin zone (i.e., $k_0$ satisfies Assumption \textbf{A5}), then $k_0=(\delta_1 \pi, \delta_2 \pi,..., \delta_d \pi)$, where $\delta_{j} \in \{0,1\}$ for $j \in \{1,...,d\}$. Hence, \mref{fund-dom} becomes:
%For localizing around $k_0$, we introduce the following fundamental domain (with respect to the dual lattice $2\pi \bZ^d$) 
$$\mathcal{O}=\prod_{j=1}^d [(\delta_{j}-1)\pi, (\delta_j+1)\pi].$$
\end{notation}
\subsection{The Floquet transforms on abelian coverings}
The following transform will play the role of the Fourier transform for the periodic case. Indeed, it is a version of the Fourier transform on the group $\mathbb{Z}^d$ of periods. 
\begin{defi}
\label{Floquet transforms}
The \textbf{Floquet transform} $\mathcal{F}$ (which depends on the choice of $h$)
\begin{equation*}
%\label{E:transform}
f(x) \rightarrow \widehat{f}(k,x)
\end{equation*}
maps a compactly supported function $f$ on $X$ into a function $\widehat{f}$ defined on $\mathbb{R}^d \times X$ in the following way:

\begin{equation*}
%\label{eqn:hat f(k)}
\widehat{f}(k,x):=\sum_{\gamma \in \mathbb{Z}^d}f(\gamma \cdot x)e^{-i h(\gamma \cdot x) \cdot k}.
\end{equation*}
\end{defi}
From the above definition, one can see that $\widehat{f}$ is $\mathbb{Z}^d$-periodic in the $x$-variable and satisfies a cyclic condition with respect to $k$:
\[   \left\{
\begin{array}
{ll}
\label{E:periodic}
      \widehat{f}(k,\gamma \cdot x)=\widehat{f}(k,x), \quad \forall \gamma \in \mathbb{Z}^d \\
      \widehat{f}(k+2\pi \gamma,x)=e^{-2\pi i \gamma \cdot h(x)}\widehat{f}(k,x), \quad \forall \gamma \in \mathbb{Z}^d \\
\end{array} 
.\right. \]

Thus, it suffices to consider the Floquet transform $\widehat{f}$ as a function defined on $\mathcal{O} \times M$. Usually, we will regard $\widehat{f}$ as a function $\widehat{f}(k,\cdot)$ in $k$-variable in $\mathcal{O}$ with values in the function space $L^{2}(M)$.

The next lemma lists some well-known results of the Floquet transform. Although the lemma is stated for abelian coverings, its proof does not require any change from the proof for the flat case. We omit the details since these can be found in \cite{K}, for instance.
\begin{lemma}
\label{L:floquet}
\begin{enumerate}[i)]
\item The transform $\mathcal{F}$ is an isometry of $L^{2}(X)$ onto 
\begin{equation*}
%\label{E:direct_integral1}
\int_{\mathcal{O}}^{\oplus}L^{2}(M)=L^{2}(\mathcal{O},L^{2}(M))
\end{equation*}
and of $H^{2}(X)$ onto 
\begin{equation*}
%\label{E:direct_integral2}
\int_{\mathcal{O}}^{\oplus}H^{2}(M)=L^{2}(\mathcal{O},H^{2}(M)).
\end{equation*}

\item The following two equivalent inversion formulae $\mathcal{F}^{-1}$ are given by
\begin{equation}
\label{E:inversion1}
f(x)=(2\pi)^{-d}\int_{\mathcal{O}}e^{ik \cdot h(x)}\widehat{f}(k,x)dk, \quad x \in X.
\end{equation}
and
\begin{equation}
\label{E:inversion2}
f(x)=(2\pi)^{-d}\int_{\mathcal{O}}e^{ik \cdot h(x)}\widehat{f}(k,\gamma^{-1} \cdot x)dk, \quad x \in \gamma \cdot \overline{F(M)}.
\end{equation}

\item The action of any periodic elliptic operator $P$ in $L^{2}(X)$ under the Floquet transform $\mathcal{F}$ is given by 
\begin{equation*}
%\label{E:conjugate_floquet}
\mathcal{F}P(x,D)\mathcal{F}^{-1}=\int_{\mathcal{O}}^{\oplus}P(k)dk,
\end{equation*}
where $P(k)(x,D)=e^{-ik\cdot h(x)}P(x,D)e^{ik\cdot h(x)}$.
In other words, 
\begin{equation*}
\widehat{Pf}(k)=P(k)\widehat{f}(k), \quad \forall f \in H^{2}(X).
\end{equation*}
\end{enumerate}
\end{lemma}

\begin{remark}
\label{direct_integral}
The direct integral decomposition of $P$ in Lemma \ref{L:floquet} (iii) has an important consequence that the spectrum of any periodic elliptic operator $P$ on $X$ is the union of the spectra of operators $P(k)$ on $M$ over the fundamental domain $\mathcal{O}$.
\end{remark}

From now on, we will consider the Green's function $G_{\lambda}(x,y)$ at the level $\lambda$ in Case I (i.e., \textit{Spectral gap interior}) in the rest of this section. 
%%%%%%%%%%%%%%%%%%%%%%%%%%%
\subsection{A Floquet reduction of the problem}
\label{fl-reduction}
We begin with the following proposition, which says roughly that if one starts moving $k$ from some shifted copy of the Brillouin zone along the direction $i\beta_s$, then $k_0+i\beta_s$ is the \textbf{first} quasimomentum $k$ that \textbf{belongs to the Fermi surface $F_{L, \lambda}$}.

\begin{prop} 
\label{singularity}
If $|\lambda|$ is small enough (depending on the dispersion branch $\lambda_j$ and $L$), then for any $(t,s) \in [0,1] \times \mathbb{S}^{d-1}$, we have $\lambda \in \sigma(L(k+it\beta_s))$ if and only if $(k,t)=(k_0,1)$.
\end{prop}
This statement is proven in \cite[Proposition 4.1]{KKR} for the flat case. The case of an abelian covering does not require any change in the proof. The main ingredients in the proof are the upper-semicontinuity of the
spectra of the analytic family $\{L(k)\}_{k \in \mathbb{C}^d}$ and the fact that $E$ is a \textbf{real} function, whose Hessian is negative definite (Lemma \ref{func_E}).

We consider the following real, smooth linear elliptic operators on $X$:
\begin{equation*}
\label{E:L_t_s}
L_{t,s}=e^{t\beta_s\cdot h(x)}Le^{-t\beta_s\cdot h(x)}, \quad (t,s) \in [0,1] \times \mathbb{S}^{d-1}.
\end{equation*}
Notice that these operators are $G$-periodic, and when pushing $L_{t,s}$ down to $M$, we get the operator $L(-i\beta_s)$. We also use the notation $L_s$ for $L_{1,s}$.

Due to Remark \ref{direct_integral}, we can apply the identity \mref{fl_spectrum} to the operator $L_{t,s}$ to obtain
\begin{equation}
\label{E:band_functions_L_t_s}
\sigma(L_{t,s})=\bigcup_{k \in \mathcal{O}} \sigma(L_{t,s}(k))=\bigcup_{k \in \mathcal{O}} \sigma(L(k+it\beta_s)) \supseteq \{\lambda_{j}(k+it\beta_s)\}_{k\in \mathcal{O}}.
\end{equation}

We now fix a real number $\lambda$ such that the statement of Proposition \ref{singularity} holds. 
By \mref{E:band_functions_L_t_s} and Proposition \ref{singularity}, $\lambda$ is in the resolvent set of $L_{t,s}$ for any $(t,s) \in [0,1) \times \mathbb{S}^{d-1}$.
Let $R_{t,s,\lambda}$ be the resolvent operator $(L_{t,s}-\lambda)^{-1}$.
Using Lemma \ref{L:floquet} (iii), for any $f \in L^{2}_{comp}(X)$, we have
\begin{equation*}
\widehat{R_{t,s,\lambda}f}(k)=(L_{t,s}(k)-\lambda)^{-1}\widehat{f}(k), \quad (t,k) \in [0,1) \times \mathcal{O}.
\end{equation*}
Due to Lemma \ref{L:floquet} (i), the sesquilinear form $(R_{t,s,\lambda}f,\varphi)$ is equal to
\begin{equation*}
(2\pi)^{-d}\int_{\mathcal{O}}\left( (L_{t,s}(k)-\lambda)^{-1}\widehat{f}(k), \widehat{\varphi}(k) \right)dk,
\end{equation*}
where $\varphi \in L^2_{comp}(X)$. 

In the next lemma, the weak convergence as $t \nearrow 1$ of the operator $R_{t,s,\lambda}$ in $L^{2}_{comp}(X)$ is proved and thus, we can introduce the limit operator $\displaystyle R_{s,\lambda}:=\lim_{t \rightarrow 1^{-}}R_{t,s,\lambda}$.

\begin{lemma}
\label{L:bilin}
Let $d\geq 2$. Under Assumption A, for $f,\varphi$ in $L^{2}_{comp}(X)$, the following equality holds:
\begin{equation}
\label{E:limit}
\lim_{t \rightarrow 1^{-}}(R_{t,s,\lambda}f,\varphi)=(2\pi)^{-d}\int_{\mathcal{O}}\left( L_{s}(k)-\lambda)^{-1}\widehat{f}(k),\widehat{\varphi}(k)\right)dk.
\end{equation}
The integral in the right hand side of \mref{E:limit} is absolutely convergent. 
\end{lemma}
This lemma is a direct corollary of Lemma \ref{L:Bloch_variety}, Proposition \ref{singularity} and the Lebesgue Dominated Convergence Theorem as being shown in \cite{KKR}. We skip the proof.

For any $(t,s) \in [0,1) \times \mathbb{S}^{d-1}$, let $G_{t,s, \lambda}$ 
be the Green's function of $L_{t,s}$ at $\lambda$, which is the kernel of $R_{t,s,\lambda}$.
Thus, $$G_{t,s,\lambda}(x,y)=e^{t\beta_s \cdot (h(x)-h(y))}G_{\lambda}(x,y).$$
Taking the limit and applying Lemma \ref{L:bilin}, we conclude that the function 
$$G_{s, \lambda}(x,y):=e^{\beta_s \cdot (h(y)-h(x))}G_{\lambda}(x,y)$$ 
is the integral kernel of the operator $R_{s,\lambda}$ defined as follows:
\begin{equation}
\label{E:R_s}
\widehat{R_{s,\lambda}f}(k)=(L_{s}(k)-\lambda)^{-1}\widehat{f}(k).
\end{equation}
Hence, the problem of finding asymptotics of $G_{\lambda}$ is now equivalent 
to obtaining asymptotics of any function $G_{s,\lambda}$, 
where $s$ is an admissible direction in $\mathcal{A}_h$. 

In addition, by \mref{E:inversion1} and  \mref{E:R_s}, the function $G_{s,\lambda}$, which is also the Green's function of the operator $L_{s}$ at $\lambda$, is the integral kernel of the operator $R_{s,\lambda}$ that acts on $L^{2}_{comp}(X)$ in the following way:
\begin{equation}
\label{E:R_s_exp}
R_{s,\lambda}f(x)=(2\pi)^{-d}\int_{\mathcal{O}}e^{ik\cdot h(x)}(L_{s}(k)-\lambda)^{-1}\widehat{f}(k,x)dk, \quad x \in X.
\end{equation}
This accomplishes Step 1 in our strategy of the proof.
%%%%%%%%%%%%%%%%%%%%%%%%%%%%%%%%
%%%%%%%%%%%%%%%%%%%%%%%%%%%%%%%%
\subsection{Isolating the leading term in $R_{s,\lambda}$ and a reduced Green's function}
\label{isolate-green}
The purpose of this part is to complete Step 2, i.e., to localize the part of the integral in \mref{E:R_s_exp}, that is responsible for the leading term of the Green's function asymptotics. 

\begin{defi}
\label{spectral_proj}
For any $z \in V$, we denote by $P(z)$ the spectral projector $\chi_{B(0, \varepsilon_0)}(L(z))$, i.e., 
\begin{equation*}
\label{E:Riesz_projector}
P(z)=-\frac{1}{2\pi i}\oint_{|\alpha|=\epsilon_0}(L(z)-\alpha)^{-1}d\alpha.
\end{equation*}
By \textbf{(P2)}, $P(z)$ projects $L^{2}(M)$ onto the eigenspace spanned by $\phi_z$. We also put $Q(z):=I-P(z)$ and denote by $R(P(z))$, $R(Q(z))$ the ranges of the projectors $P(z)$, $Q(z)$ correspondingly.
\end{defi}

Using \textbf{(P6)} and the fact that $P(k+i\beta_s)^*=P(k-i\beta_s)$, 
we can deduce that if $|k-k_0| \leq r_0$ (see \mref{E:beta_s_in_V}), the following equality holds
\begin{equation}
\label{E:form_P_s}
P(k+i\beta_s)u=\frac{(u, \phi_{k-i\beta_s})_{L^{2}(M)}}{(\phi_{k+i\beta_s}, \phi_{k-i\beta_s})_{L^{2}(M)}}\phi_{k+i\beta_s}, 
\quad \forall u \in L^{2}(M).
\end{equation}

Let $\eta$ be a cut-off smooth function on $\mathcal{O}$ supported on $\{k \in \mathcal{O} \mid |k-k_0|<r_0\}$ and equal to $1$ around $k_0$.  

According to \mref{E:R_s_exp}, for any $f \in C^{\infty}_{c}(X)$, we want to find $u$ such that 
$$(L_s(k)-\lambda)\widehat{u}(k)=\widehat{f}(k).$$
Then the Green's function $G_{s, \lambda}$ satisfies
$$\int_{X}G_{s, \lambda}(x,y)f(y)d\mu_{X}(y)=\mathcal{F}^{-1}\widehat{u}(k,x)=u(x),$$
where $\mathcal{F}$ is the Floquet transform introduced in Definition \ref{L:floquet}.

By Proposition \ref{singularity}, the operator $L_s(k)-\lambda$ is invertible for any $k$ such that $k \neq k_0$. Hence, we can decompose $\widehat{u}(k)=\widehat{u_0}(k)+(L_s(k)-\lambda)^{-1}(1-\eta(k))\widehat{f}(k)$, where $\widehat{u_0}$ satisfies the equation
$$(L_s(k)-\lambda)\widehat{u_0}(k)=\eta(k)\widehat{f}(k).$$

Observe that $R(P(z))$ and $R(Q(z))$ are invariant subspaces for the operator $L(z)$ for any $z \in V$. Thus, if $u_1, u_2$ are functions such that $\widehat{u_1}(k)=P(k+i\beta_s)\widehat{u_0}(k)$ and $\widehat{u_2}(k)=Q(k+i\beta_s)\widehat{u_0}(k)$, we must have 
\begin{equation}
\label{u1}
(L_s(k)-\lambda)P(k+i\beta_s)\widehat{u_1}(k)=\eta(k)P(k+i\beta_s)\widehat{f}(k)
\end{equation}
and
$$(L_s(k)-\lambda)Q(k+i\beta_s)\widehat{u_2}(k)=\eta(k)Q(k+i\beta_s)\widehat{f}(k).$$

Due to \textbf{(P2)}, when $k$ is close to $k_0$, $\lambda=\lambda_{j}(k_0+i\beta_s)$ must belong to the resolvent of the operator $L_{s}(k)|_{R(Q(k+i\beta_s))}$.
Hence, we can write $\widehat{u_2}(k)=\eta(k)(L_s(k)-\lambda)^{-1}Q(k+i\beta_s)\widehat{f}(k)$. 
Therefore, $\widehat{u}(k)$ equals
$$\widehat{u_1}(k)+\left((1-\eta(k))(L_s(k)-\lambda)^{-1}+\eta(k)((L_s(k)-\lambda)|_{R(Q(k+i\beta_s))})^{-1}Q(k+i\beta_s)\right)\widehat{f}(k).$$

The next theorem shows that for finding the asymptotics, we can ignore the infinite-dimensional part of the operator $R_{s, \lambda}$, i.e., the second term in the above sum of two operators.
\begin{thm}
\label{microlocal}
Define
$$T_{s}(k):=(1-\eta(k))(L_s(k)-\lambda)^{-1}+\eta(k)((L_s(k)-\lambda)|_{R(Q(k+i\beta_s))})^{-1}Q(k+i\beta_s).$$

Let $T_{s}$ be the operator acting on $L^2(X)$ as follows:
$$T_s=\mathcal{F}^{-1}\left((2\pi)^{-d}\int^{\oplus}_{\mathcal{O}}T_{s}(k)dk\right)\mathcal{F}.$$

Then the Schwartz kernel $K_{s}(x,y)$ of the operator $T_{s}$ is continuous away from the diagonal of $X$, and moreover, it is also rapidly decaying in a uniform way with respect to $s \in \mathbb{S}^{d-1}$, 
i.e., for any $N>0$,
$$\sup_{s \in \mathbb{S}^{d-1}}|K_{s}(x,y)|=O(d_X(x,y)^{-N}).$$ 
\end{thm}
A proof using microlocal analysis will be mentioned in Section 7.

Now let $V_{s}:=R_{s, \lambda}-T_{s}$. Then
the Schwartz kernel $G_0(x,y)$ of the operator $V_{s}$ satisfies the following relation:
\begin{equation}
\label{G_0}
\int_X G_0(x,y)f(y)d\mu_{X}(y)=\mathcal{F}^{-1}\widehat{u_1}(k,x)=u_1(x).
\end{equation}
In what follows, we will find an integral representation of the kernel $G_0$. We will see that $G_0$ provides the leading term of the asymptotics of the kernel $G_{s, \lambda}$. For this reason, $G_0$ is called \textbf{a reduced Green's function}.

To find $u_1$, we use the equation \mref{u1} and apply \mref{E:form_P_s} to deduce
\begin{equation*}
(\lambda_{j}(k+i\beta_s)-\lambda)(\widehat{u_1}(k),\phi_{k-i\beta_s})_{L^{2}(M)}=\eta(k)(\widehat{f}(k),\phi_{k-i\beta_s})_{L^{2}(M)}.
\end{equation*}
Using $\widehat{u_1}(k)=P(k+i\beta_s)\widehat{u_1}(k)$ and \mref{E:form_P_s} again, the above identity becomes
\begin{equation*}
\label{E:def_u1}
\widehat{u_1}(k,x):=\frac{\eta(k)\phi_{k+i\beta_s}(x)(\widehat{f}(k),\phi_{k-i\beta_s})_{L^{2}(M)}}{(\phi_{k+i\beta_s},\phi_{k-i\beta_s})_{L^2(M)}(\lambda_j(k+i\beta_s)-\lambda)}, \quad k \neq k_0.
\end{equation*}
By the inverse Floquet transform \mref{E:inversion1}, for any $x \in X$,
\begin{equation*}
u_{1}(x)=(2\pi)^{-d}\int_{\mathcal{O}}e^{ik\cdot h(x)}\frac{\eta(k)\phi_{k+i\beta_s}(x)(\widehat{f}(k),\phi_{k-i\beta_s})_{L^{2}(M)}}{(\phi_{k+i\beta_s},\phi_{k-i\beta_s})_{L^2(M)}(\lambda_j(k+i\beta_s)-\lambda)}dk.
\end{equation*}

Now we repeat some calculations in \cite{KKR, KR} to have
\begin{equation*}
\begin{split}
u_1(x)&=\frac{1}{(2\pi)^{d}}\int_{\mathcal{O}}\int_{M}\frac{e^{ik\cdot h(x)}\eta(k)\widehat{f}(k,y)\overline{\phi_{k-i\beta_s}(y)}\phi_{k+i\beta_s}(x)}{(\phi_{k+i\beta_s},\phi_{k-i\beta_s})_{L^2(M)}(\lambda_j(k+i\beta_s)-\lambda)}d\mu_{M}(y)dk\\
&=\frac{1}{(2\pi)^{d}}\int_{\mathcal{O}}\int_{\overline{F(M)}}\sum_{\gamma \in G}\frac{e^{ik\cdot (h(x)-h(\gamma^{-1}\cdot y))}\eta(k)\overline{\phi_{k-i\beta_s}(y)}\phi_{k+i\beta_s}(x)}{(\phi_{k+i\beta_s},\phi_{k-i\beta_s})_{L^2(M)}(\lambda_j(k+i\beta_s)-\lambda)}d\mu_{X}(y)dk\\
&=\frac{1}{(2\pi)^{d}}\int_{\mathcal{O}}\sum_{\gamma \in G}\int_{\gamma \cdot \overline{F(M)}}f(y)\frac{e^{ik\cdot (h(x)-h(y))}\eta(k)\overline{\phi_{k-i\beta_s}(\gamma^{-1} \cdot y)}\phi_{k+i\beta_s}(x)}{(\phi_{k+i\beta_s},\phi_{k-i\beta_s})_{L^2(M)}(\lambda_j(k+i\beta_s)-\lambda)}d\mu_{X}(y)dk\\
&=\frac{1}{(2\pi)^{d}}\int_{X}f(y)\left(\int_{\mathcal{O}}\frac{e^{ik\cdot (h(x)-h(y))}\eta(k)\overline{\phi_{k-i\beta_s}(y)}\phi_{k+i\beta_s}(x)}{(\phi_{k+i\beta_s},\phi_{k-i\beta_s})_{L^2(M)}(\lambda_j(k+i\beta_s)-\lambda)}dk\right)d\mu_{X}(y).
\end{split}
\end{equation*}
In the second equality above, we use the identity \mref{integralF(M)}.

Consequently, from \mref{G_0}, we conclude that our reduced Green's function is
\begin{equation}
\label{E:formula_G0}
G_0(x,y)=\frac{1}{(2\pi)^{d}}\int_{\mathcal{O}}e^{ik\cdot (h(x)-h(y))}\eta(k)\frac{\phi_{k+i\beta_s}(x)\overline{\phi_{k-i\beta_s}(y)}}{(\phi_{k+i\beta_s},\phi_{k-i\beta_s})_{L^2(M)}(\lambda_j(k+i\beta_s)-\lambda)}dk.
\end{equation}
%%%%%%%%%%%%%%%%%%%%%%

%%%%%%%%%%%%%%%%%%%%%%%%%%%%%%%%%%%%%%%%%%%%%%%
\section{Some auxiliary statements}
In this part, we provide the analogs of some results from \cite{KKR, KR}, which do not require any significant change in the proofs when dealing with the case of abelian coverings. Instead of repeating the details, we will make some brief comments about the main ingredients of these results.

The first result studies the local smoothness in $(z,x)$ of the eigenfunctions $\phi_z(x)$ of the operator $L(z)$ with the eigenvalue $\lambda_j(z)$.
\begin{lemma}
\label{L:joint_continuity}
Suppose that $B \subset \mathbb{R}^d$ is the open ball centered at $k_0$  with radius $r_0$ (see \mref{E:beta_s_in_V}). Then for each $s \in \mathbb{S}^{d-1}$, the functions $\displaystyle \phi_{k\pm i\beta_s}(x)$ are smooth on a neighborhood of $\overline{B} \times M$ in $\mathbb{R}^d \times M$. In addition, for any multi-index $\alpha$, the functions $\displaystyle D^{\alpha}_{k}\phi_{k\pm i\beta_s}(x)$ are also jointly continuous in $(s,k,x)$. In particular, we have
\begin{equation*}
\sup_{(s,k,x) \in \mathbb{S}^{d-1}\times \overline{B}  \times M} |D^{\alpha}_{k}\phi_{k \pm i\beta_s}(x)|<\infty.
\end{equation*}
\end{lemma}
To obtain Lemma \ref{L:joint_continuity}, one can modify the proof of \cite[Proposition 9.6]{KKR} without any significant change.
Indeed, the three main ingredients in the proof are the smoothness in $z$ of the family of operators $\{L(z)\}_{z \in V}$ acting between Sobolev spaces (Lemma \ref{fiber_op_structure}), the property \textbf{(P3)} for bootstraping regularity of eigenfunctions in $k$, and the standard coercive estimates of elliptic operators $L(z)$ on the compact manifold $M$ (see e.g., \cite[estimate (11.29)]{Taylor1}) for bootstraping regularity in $x$.

The next result is the asymptotics of the scalar integral expression obtained from the integral representation \mref{E:formula_G0} of the reduced Green's function $G_0$.
\begin{prop}
\label{P:asymp_KKR}
Suppose that $d\geq 2$ and $B$ is the open ball defined in Proposition \ref{L:joint_continuity}. Let $\eta(k)$ be a smooth cut off function around the point $k_0$, and $\{\mu_{s}(k,x,y)\}_{s \in \mathbb{S}^{d-1}}$ be a family of smooth $\mathbb{C}^d$-valued functions defined on $\overline{B} \times M \times M$. We also use the same notation $\mu_{s}(k,x,y)$ for its lift to $\overline{B} \times X \times X$.
For each quadruple $(s,a,x,y) \in \mathbb{S}^{d-1} \times \mathbb{R}^d \times X \times X$, we define
$$I(s,a):=\frac{1}{(2\pi)^d}\int_{\mathcal{O}}e^{ik\cdot a}\frac{\eta(k)}{\lambda_j(k+i\beta_s)-\lambda}dk$$
and
$$J(s,a, x, y):=\frac{1}{(2\pi)^d}\int_{\mathcal{O}}e^{ik\cdot a}\frac{\eta(k)(k-k_0)\cdot \mu_s(k,x,y)}{\lambda_j(k+i\beta_s)-\lambda}dk.$$
Assume that the size of the support of $\eta$ is small enough. Fix a direction $s \in \mathbb{S}^{d-1}$ and consider all vectors $a$ such that $\displaystyle s=\frac{a}{|a|}$. Then when $|a|$ is large enough, we have
\begin{equation}
\label{E:asymp_I}
I(s,a)=\frac{e^{ik_0 \cdot a}|\nabla E(\beta_s)|^{(d-3)/2}}{(2\pi|a|)^{(d-1)/2}\det{(-\mathcal{P}_s \Hess{(E)}(\beta_s)\mathcal{P}_s)}^{1/2}}+O(|a|^{-d/2})
\end{equation}
and 
\begin{equation}
\label{E:asymp_J}
\sup_{(x,y) \in X \times X}|J(s,a,x,y)|=O(|a|^{-d/2}).
\end{equation}
Moreover, if all derivatives of $\mu_s(k,x,y)$ with respect to $k$ are uniformly bounded in $s \in \mathbb{S}^{d-1}$,
then all the terms $O(\cdot)$ in \mref{E:asymp_I} and \mref{E:asymp_J} are also uniform in $s \in \mathbb{S}^{d-1}$ when $|a| \rightarrow \infty$.
\end{prop}
The proof of Proposition \ref{P:asymp_KKR} can be extracted from \cite[Section 6]{KKR}. The main ingredient (see \cite[Proposition 6.1]{KKR}) is an application of the Weierstrass Preparation Lemma in several complex variables to have a factorization of the denominator $\lambda_j(k+i\beta_s)-\lambda$ of the integrands of $I, J$ into a form that is close to the normal form in the free case. This trick was used in \cite{Woess} in the discrete setting.

The next result \cite[Theorem 3.3]{KR} will be needed in the proof of Theorem \ref{main_KR}.
\begin{prop}
\label{P:asymp_KR}
Assume $d \geq 3$. Let $a \in \mathbb{R}^d$. Let $\eta$ be a smooth function satisfying the assumptions of Proposition \ref{P:asymp_KKR}, and let $\mu(k,x,y)$ be a smooth $G$-periodic function from a neighborhood of $\overline{B} \times X \times X$ to $\mathbb{C}^d$. Then the following asymptotics hold when $|a| \rightarrow \infty$:
\begin{equation*}
\label{E:asymp_KR}
\begin{split}
&\frac{1}{(2\pi)^{d}}\int_{\mathcal{O}}e^{ik\cdot a}\frac{\eta(k)}{\lambda_j(k)}dk=\frac{\Gamma(\frac{d}{2}-1)e^{ik_0 \cdot a}}{2\pi^{d/2}(\det{H})^{1/2}|H^{-1/2}(a)|^{d-2}}(1+O(|a|^{-1}),\\
\mbox{and}&\\
&\sup_{x,y \in X}\left|\int_{\mathcal{O}}e^{ik\cdot a}\frac{\eta(k)(k-k_0)\cdot \mu(k,x,y)}{\lambda_j(k)}dk\right|=O(|a|^{-d+1}).
\end{split}
\end{equation*} 
\end{prop}
%%%%%%%%%%%%%%%%%%%%%%%%%%%%%%%%%%%%%%%%%%%%%%%
\section{Proofs of the main results}
\begin{prmainKKR}
We fix an admissible direction $s$ of the additive function $h$ and consider any $x,y \in X$ such that 
$$\displaystyle \frac{h(x)-h(y)}{|h(x)-h(y)|}=s \in \mathcal{A}_h.$$

As we discussed in Section 4, the Green's function $G_{\lambda}$ satisfies
\begin{equation}
\label{G_s}
G_{\lambda}(x,y)=e^{\beta_s \cdot (h(y)-h(x))}G_{s, \lambda}(x,y),
\end{equation}
where $G_{s, \lambda}$ is the Schwartz kernel of the resolvent operator $R_{s}$. Also, $R_{s, \lambda}=V_{s}+T_{s}$. Due to Theorem \ref{microlocal}, the Schwartz kernel of $T_{s}$ decays rapidly (uniformly in $s$) when $d_X(x,y)$ is large enough. Hence, to find the asymptotics of the kernel of $R_{s, \lambda}$, it suffices to consider the kernel $G_0$ of the operator $V_{s}$. 
Define 
\begin{equation}
\label{a}
a:=h(x)-h(y)
\end{equation}
and
$$\tilde{\mu}_{\omega}(k,p,q):=\frac{\phi_{k+i\beta_{\omega}}(p)\overline{\phi_{k-i\beta_{\omega}}(q)}}{(\phi_{k+i\beta_{\omega}},\phi_{k-i\beta_{\omega}})_{L^2(M)}}, \quad (\omega,p,q) \in \mathbb{S}^{d-1} \times M \times M.$$
By Lemma \ref{L:joint_continuity}, $\tilde{\mu}_{\omega}$ is a smooth function on $\overline{B} \times M \times M$. By Taylor expanding around $k_0$, $\tilde{\mu}_{\omega}(k,p,q)=\tilde{\mu}_{\omega}(k_0,p,q)+(k-k_0)\cdot \mu_{\omega}(k,p,q)$ for some smooth $\mathbb{C}^d$-valued function $\mu_{\omega}(k,p,q)$ defined on $\overline{B} \times M \times M$. From Lemma \ref{L:joint_continuity} and the definition of $\tilde{\mu_{\omega}}$, 
$$\sup_{(\omega,k,x,y) \in  \mathbb{S}^{d-1} \times\overline{B} \times M \times M} |D^{\alpha}_{k}\tilde{\mu}_{\omega}(k,x,y)|<\infty,$$
for any multi-index $\alpha$. 
Thus, all derivatives of $\mu_{\omega}$ with respect to $k$ are also uniformly bounded in $\omega \in \mathbb{S}^{d-1}$.
 We now can rewrite \mref{E:formula_G0} as follows:
\begin{equation*}
\label{E:formula_G0_1}
\begin{split}
G_0(x,y)&=\frac{1}{(2\pi)^{d}}\int_{\mathcal{O}}e^{ik\cdot a}\frac{\eta(k)}{\lambda_j(k+i\beta_s)-\lambda}\left(\tilde{\mu}_s(k_0,x,y)+(k-k_0)\cdot \mu_s(k,x,y)\right)dk\\
&=I(s,a)\frac{\phi_{k_0+i\beta_s}(x)\overline{\phi_{k_0-i\beta_s}(y)}}{(\phi_{k_0+i\beta_s},\phi_{k_0-i\beta_s})_{L^2(M)}}+J(s,a, x, y).
\end{split}
\end{equation*}
Here the integrals $I(s,a)$ and $J(s,a,x,y)$ are defined in Proposition \ref{P:asymp_KKR}. By using Proposition \ref{P:asymp_KKR}, we obtain the following asymptotics whenever $|a|$ is large enough:
\begin{equation}
\label{E:formula_G0_2}
\begin{split}
G_0(x,y)&=\Big(\frac{e^{ik_0 \cdot a}|\nabla E(\beta_s)|^{(d-3)/2}}{(2\pi|a|)^{(d-1)/2}\det{(-\mathcal{P}_s \Hess{(E)}(\beta_s)\mathcal{P}_s)}^{1/2}}+O(|a|^{-d/2})\Big)\\
&\times \frac{\phi_{k_0+i\beta_s}(x)\overline{\phi_{k_0-i\beta_s}(y)}}{(\phi_{k_0+i\beta_s},\phi_{k_0-i\beta_s})_{L^2(M)}}+O(|a|^{-d/2}),
\end{split}
\end{equation}
where all the terms $O(\cdot)$ are uniform in $s$. Due to \mref{a} and Proposition \ref{P:add_func}, $O(|a|^{\ell})=O(d_X(x,y)^{\ell})$ for any $\ell \in \bZ$, provided that $d_X(x,y)>R_h$. Hence, by choosing the constant $R_h$ larger if necessary, we can assume that when $d_X(x,y)>R_h$, the asymptotics \mref{E:formula_G0_2} would follow.
Finally, we substitute \mref{a} to the asymptotics \mref{E:formula_G0_2} and then use \mref{G_s} to deduce Theorem \ref{main}.
\end{prmainKKR}
\vspace{0.3cm}
\begin{prmainKR}
We recall that $\lambda=\lambda_j(k_0)=0$ and $R_{-\varepsilon}$ is the resolvent operator $(L+\varepsilon)^{-1}$ when $\varepsilon>0$ is small enough. We will repeat the Floquet reduction approach in Section 4. Given any $f, \varphi \in L^2_{comp}(X)$, the sesquilinear form $\langle R_{-\varepsilon}f, \varphi \rangle$ is
$$(2\pi)^{-d}\int_{\mathcal{O}}\left((L(k)+\varepsilon)^{-1}\widehat{f}(k), \widehat{\varphi}(k)\right)dk.$$
The first conclusion of this theorem is achieved by a smilar argument in \cite[Lemma 2.3]{KR}. Hence the operator $R=\lim_{\varepsilon \rightarrow 0^+}R_{-\varepsilon}$ is defined by the identity $\widehat{Rf}(k)=R(k)\widehat{f}(k)$ and the Green's function $G$ is the Schwartz kernel of the operator $R$. To single out the principal term in $R$, we first choose a neighborhood $V \subset \mathcal{O}$ of $k_0$ such that when $k \in V$, there is a non-zero $G$-periodic eigenfunction $\phi_k(x)$ of the operator $L(k)$ with the corresponding eigenvalue $\lambda_j(k)$ and moreover, the mapping $k \mapsto \phi_k(\cdot)$ is analytic in $k$ as a $H^2(M)$-valued function. This is always possible due to Lemma \ref{L:Bloch_variety}. For such $k \in V$,  let us denote by $P(k)$ the spectral projector of $L(k)$ that projects $L^2(M)$ onto the eigenspace spanned by $\phi_k$. The notation $R(I-P(k))$ stands for the range of the projector $I-P(k)$.
Then we pick $\eta$ as a smooth cut off function around $k_0$ such that $\supp(\eta) \Subset V$.
Define the operator
$$T:=\frac{1}{(2\pi)^d}\int^{\oplus}_{\mathcal{O}}T(k)dk,$$
where
$$T(k):=(1-\eta(k))L(k)^{-1}+\eta(k)(L(k)|_{R(I-P(k))})^{-1}(I-P(k)).$$
As in Theorem \ref{microlocal}, the Schwartz kernel $K(x,y)$ of $T$ is rapidly decaying as $d_X(x,y) \rightarrow \infty$. Thus, the asymptotics of the Green's function $G$ are the same as the asymptotics of the Schwartz kernel $G_0$ of the operator $R-T$. To find $G_0$, we repeat the arguments in Subsection 4.3 to derive the formula
$$G_0(x,y)=\frac{1}{(2\pi)^{d}}\int_{\mathcal{O}}e^{ik\cdot (h(x)-h(y))}\frac{\eta(k)}{\lambda_j(k)}\frac{\phi_{k}(x)\overline{\phi_{k}(y)}}{\|\phi_{k}\|^2_{L^2(M)}}dk, \quad x,y \in X.$$
As in the proof of Theorem \ref{main}, we set $a:=h(x)-h(y)$ and rewrite the smooth function 
$$\frac{\phi_{k}(x)\overline{\phi_{k}(y)}}{\|\phi_{k}\|^2_{L^2(M)}}=\frac{\phi_{k_0}(x)\overline{\phi_{k_0}(y)}}{\|\phi_{k_0}\|^2_{L^2(M)}}+(k-k_0)\cdot \mu(k,x,y),$$
for some smooth $G$-periodic function $\mu: \overline{B}\times X \times X \rightarrow \mathbb{C}^d$. 
Now by applying Proposition \ref{P:asymp_KR}, the proof is completed.
\end{prmainKR}
%%%%%%%%%%%%%%%%%%%%%%%
\section{Proofs of technical statements}
\subsection{Proofs of Proposition \ref{P:add_func} and Proposition \ref{A_h}}
\begin{praddfunc}
Fixing a point $x_0 \in X$, we let 

$$K:=\overline{F(M)},$$

$$R:=\max_{x \in K} d_X(x_0,x),$$

and 
$$\widetilde{R_h}:=\max_{(x,y) \in K \times K}|h(x)-h(y)|.$$ 

Due to Proposition \ref{P:quasi_iso} and the fact that $|\cdot|_S$ is equivalent to $|\cdot|$ on $\bZ^d$, there exist $C_1>1$ and $C_2>0$ such that
$$C_1^{-1}\cdot d_X(g_1 \cdot x_0, g_2 \cdot x_0)-C_2 \leq |g_1-g_2| \leq C_1\cdot d_X(g_1 \cdot x_0,g_2 \cdot x_0)+C_2,$$
for any $g_i \in \bZ^d$, $i=1,2$.

Now we consider any two points $x,y$ in $X$. By \mref{closure_F(M)}, we can select $\tilde{x}$, $\tilde{y}$ in $K$ such that $x=g_1 \cdot \tilde{x}$ and $y=g_2 \cdot \tilde{y}$ for some $g_1, g_2 \in \bZ^d$.
Since $\bZ^d$ acts by isometries, we get 
\begin{equation}
\label{E:isometry}
d_X(g_1 \cdot x_0, g_1 \cdot \tilde{x})=d_X(x_0,\tilde{x}) \quad \mbox{and} \quad d_X(g_2 \cdot x_0, g_2 \cdot \tilde{y})=d_X(x_0,\tilde{y}).
\end{equation}
By \mref{additivity}, we have
$$h(x)-h(y)=h(\tilde{x})-h(\tilde{y})+g_1-g_2.$$
Using triangle inequalities and \mref{E:isometry}, we obtain
\begin{equation*}
\begin{split}
|h(x)-h(y)| &\leq \widetilde{R_h}+|g_1-g_2| \leq C_1\cdot d_X(g_1 \cdot x_0, g_2 \cdot x_0)+\widetilde{R_h}+C_2
\\&\leq C_1 \cdot d_X(x,y)+C_1 \cdot(d_X(x_0, \tilde{x})+d_X(x_0, \tilde{y}))+\widetilde{R_h}+C_2
\\&\leq C_1 \cdot d_X(x,y)+(2C_1 R+\widetilde{R_h}+C_2).
\end{split}
\end{equation*}
Likewise, 
\begin{equation*}
\begin{split}
|h(x)-h(y)| &\geq |g_1-g_2|-\widetilde{R_h} \geq C_1^{-1} \cdot d_X(g_1 \cdot x_0, g_2 \cdot x_0)-(\widetilde{R_h}+C_2)
\\&\geq C_1^{-1} \cdot d_X(x,y)-(C_1^{-1} \cdot(d_X(x_0, \tilde{x})+d_X(x_0, \tilde{y}))+\widetilde{R_h}+C_2)
\\&\geq C_1^{-1} \cdot d_X(x,y)-(2C_1 R+\widetilde{R_h}+C_2).
\end{split}
\end{equation*}
The statement follows if we put $C:=2C_1$ and $R_h:=2C_1 (2C_1 R+\widetilde{R_h}+C_2)$.
\end{praddfunc}
%%%%%%%%%%%%%%%%%%%%%%%%%%%%
\vspace{5pt}
\begin{prA_h}
By Definition \ref{Add_func}, any rational point in the unit sphere $\mathbb{S}^{d-1}$ is an admissible direction of the additive function $h$ and thus we have \mref{rational}. By using the stereographic projection, one can see that the subset $\mathbb{Q}^{d}\cap \mathbb{S}^{d-1}$ is dense in $\mathbb{S}^{d-1}$. Hence, the density of $\mathcal{A}_h$ follows.
Now we consider the case $d=2$.
For any point $x_0 \in X$, we denote by $\mathcal{A}_h(x_0)$ the subset of $\mathcal{A}_h$ consisting of unit vectors $s$ such that there exists a point $x$ in $\{ x \in X \mid d_X(x,x_0)>R_h\}$ satisfying either
$h(x)-h(x_0)=|h(x)-h(x_0)|s$
or
$h(x_0)-h(x)=|h(x)-h(x_0)|s.$
It is enough to prove that for any $x_0$, $\mathcal{A}_h(x_0)=\mathbb{S}^1$. Without loss of generality, we suppose that $h(x_0)=0$. Let $Y$ be the range of the continuous function $\displaystyle x \mapsto \frac{h(x)}{|h(x)|}$, which is defined on the connected set $\{ x \in X \mid d_X(x,x_0)>R_h\}$.
Then $Y$ is a connected subset that contains $\mathbb{Q}^2 \cap \mathbb{S}^1$ since $h(n\cdot x_0)=n$ for any $n \in \bZ^d$. Suppose for contradiction, there is a unit vector $s$ such that $s \notin \mathcal{A}_h(x_0)$ and hence, $Y \subseteq \mathbb{S}^1 \setminus \{\pm s\}$. Thus, $Y$ cannot be connected, which is a contradiction.
\end{prA_h}
%For any $s \in \mathbb{S}^{1}$, let $L_s$ be the straight line in $\mathbb{R}^d$ that goes through the points $0$ and $s$.
%Due to Proposition \ref{P:add_func}, we know that $h(x) \neq 0$ if $d_X(x,x_0)>R_h$.
%Hence, a unit vector $s$ does not belong to the set $\mathcal{A}_h(x_0)$ if and only if $h(\{x \in X \mid d_X(x,x_0)>R_h\}) \cap L_s=\emptyset$. 
%
%Suppose for contradiction, there is a vector $s$ in the unit circle such that $s \notin \mathcal{A}_h(x_0)$. Since the set $h(\{x \in X \mid d_X(x,x_0)>R_h\})$ is connected, both the condition $h(\{x \in X \mid d_X(x,x_0)>R_h\}) \cap L_s=\emptyset$ and $\mathbb{Q}^2 \cap \mathbb{S}^1 \subset h(\{x \in X \mid d_X(x,x_0)>R_h\})$ would yield a contradiction!

%%%%%%%%%%%%%%%%%%%%%%%%%%%%
\subsection{Proof of Theorem \ref{microlocal}}
It suffices to prove the following claim:
\begin{thm}
\label{microlocal_2}
Let $\phi$ and $\theta$ be two functions in $C^{\infty}_c(X)$ such that the metric distance on $X$ between the supports of these two functions is bigger than $R_h$. Let $K_{s,\phi,\theta}$ be the Schwartz kernel of the operator $\phi T_s \theta$. 
Then $K_{s,\phi,\theta}$ is continuous and rapidly decaying (uniformly in $s$) on $X \times X$, i.e., for any $N>0$, we have
$$\sup_{s \in \mathbb{S}^{d-1}}|K_{s,\phi,\theta}(x,y)| \leq C (1+d_X(x,y))^{-N},$$
for some positive constant $C=C(N, \|\phi\|_{\infty}, \|\theta\|_{\infty})$.
\end{thm}
Let $K_s(k,x,y)$ be the Schwartz kernel of the operator $T_s(k)$. 
The next lemma is an analog for abelian coverings of \cite[Lemma 7.15]{KKR}.
\begin{lemma}
\label{kernel_formula}
Let $\phi$ and $\theta$ be any two compactly supported functions on $X$ such that $\supp(\phi) \cap \supp(\theta) =\emptyset$. Then the following identity holds for any $(x,y) \in X \times X$:
\begin{equation*}
K_{s,\phi,\theta}(x,y)=\frac{1}{(2\pi)^d}\int_{\mathcal{O}} e^{ik\cdot (h(x)-h(y))}\phi(x)K_s(k,\pi(x),\pi(y))\theta(y)dk,
\end{equation*}
where $\pi$ is the covering map $X \rightarrow M$.
\end{lemma}

\begin{proof}
Let $\mathcal{P}$ be the subset of $C^{\infty}_c(X)$ consisting of all functions $\psi$ whose support is connected, and if $\gamma \in G$ such that $\displaystyle \supp{\psi^{\gamma}} \cap \supp{\psi} \neq \emptyset$ then $\gamma$ is the identity element of the deck group $G$. Since any compactly supported function on $X$ can be decomposed as a finite sum of functions in $\mathcal{P}$, we can assume that both $\phi$ and $\theta$ belong to $\mathcal{P}$. Then the rest is similar to the proof of \cite[Lemma 7.15]{KKR}.
\end{proof}

Another key ingredient in proving Theorem \ref{microlocal_2} is the following result:
\begin{prop}
\label{kernel_estimate}
Let $\displaystyle \dim{M}=n$. Then for any multi-index $\alpha$ such that $|\alpha|\geq n$, $D^{\alpha}_k K_s(k,x,y)$ is a continuous function on $M \times M$. Furthermore, we have
\begin{equation*}
\sup_{(s,k,x,y) \in \mathbb{S}^{d-1} \times \mathcal{O} \times M \times M}|D_{k}^{\alpha}K_s(k,x,y)|< \infty.
\end{equation*}
\end{prop}

Before providing the proof of Proposition \ref{kernel_estimate}, let us use it to prove Theorem \ref{microlocal_2}.
\begin{prmicrolocal}
The exponential function $e^{2\pi i\gamma \cdot h(x)}$ is $G$-periodic for any $\gamma \in G$, and hence, it is also defined on $M$. We use the same notation $e^{2\pi i\gamma \cdot h(x)}$ for the corresponding multiplication operator on $L^2(M)$. Then we can write
$$T_s(k+2\pi \gamma)=e^{-2\pi i\gamma \cdot h(x)}T_s(k)e^{2\pi i\gamma \cdot h(x)}, \quad (k,\gamma) \in \mathcal{O} \times G$$
It follows that for any multi-index $\alpha$,
\begin{equation}
\label{perboundary}
e^{i(k+2\pi \gamma)\cdot(h(x)-h(y))}\nabla_{k}^{\alpha}K_s(k+2\pi \gamma,\pi(x),\pi(y))=e^{ ik\cdot(h(x)-h(y))}\nabla_{k}^{\alpha}K_s(k,\pi(x),\pi(y)).
\end{equation}
Now we apply integration by parts to the identity in Lemma \ref{kernel_formula} to obtain
\begin{equation}
\label{int_by_parts}
 i^{N}(h(x)-h(y))^{\alpha}K_{s,\phi,\theta}(x,y)=\frac{\phi(x)\theta(y)}{(2\pi)^d}\int_{\mathcal{O}} e^{ik\cdot (h(x)-h(y))}\nabla^{\alpha}_{k} K_s(k,\pi(x), \pi(y))dk.
\end{equation}
Note that due to \mref{perboundary}, when using integration by parts, we do not have any boundary term. 
If $|\alpha|\geq n$, then the above integral is uniformly bounded in $(s,x,y)$ by Proposition \ref{kernel_estimate}. 
When $\phi(x)\theta(y) \neq 0$, we have $d_X(x,y)>R_h$ and so, $h(x) \neq h(y)$ by Proposition \ref{P:add_func}.
Therefore, the kernel $K_{s,\phi, \theta}(x,y)$ is continuous on $X\ \times X$. Now fix $(x,y)$ such that $\phi(x)\theta(y) \neq 0$. Next we choose $\ell_0 \in \{1, \dots, d\}$ such that $|h_{\ell_0}(x)-h_{\ell_0}(y)|=\max_{1 \leq \ell \leq d}|h_{\ell}(x)-h_{\ell}(y)|>0$. Fix any $N \geq n$. Let $\alpha=(\alpha_1,\dots,\alpha_d)=N(\delta_{1,\ell_0},\dots,\delta_{d,\ell_0})$, where $\delta_{\cdot,\cdot}$ is the Kronecker delta. Then $|(h(x)-h(y))^{\alpha}|^{-1}=|h_{\ell_0}(x)-h_{\ell_0}(y)|^{-N} \leq d^{N/2}|h(x)-h(y)|^{-N}$. Consequently, from \mref{int_by_parts}, we derive a positive constant $C$ (independent of $x,y$) such that
$$\sup_{s \in \mathbb{S}^{d-1}}|K_{s,\phi, \theta}(x,y)| \leq C|\phi(x)\theta(y)| |(h(x)-h(y))^{\alpha}|^{-1}\leq  Cd^{N/2}\|\phi\|_{\infty}\|\theta\|_{\infty} |h(x)-h(y)|^{-N}.$$
Using Proposition \ref{P:add_func}, the above estimate becomes
\begin{equation*}
\sup_{(s,x,y) \in \mathbb{S}^{d-1}\times X \times X}(1+d_X(x,y))^{N}|K_{s,\phi, \theta}(x,y)|<\infty,
\end{equation*}
which yields the conclusion.
\end{prmicrolocal}

Back to Proposition \ref{kernel_estimate}, we first introduce several notions.
Let $\mathcal{S}(M)$ be the space of Schwartz functions on $M$. The first notion is about the order of an operator on the Sobolev scale (see e.g., \cite[Definition 5.1.1]{Ruzh-Turu}).
\begin{defi}
\label{order_sobolev}
A linear operator $A: \mathcal{S}(M) \rightarrow \mathcal{S}(M)$ is said to be \textbf{of order $\ell \in \mathbb{R}$ on the Sobolev scale $(H^{m}(M))_{m \in \mathbb{R}}$} if for every $m \in \mathbb{R}$ it can be extended to a bounded linear operator $A_{m,m-\ell} \in B(H^{m}(M), H^{m-\ell}(M))$. In this situation, we denote by the same notation $A$ any of the operators $A_{m,m-\ell}$.

A typical example of an operator of order $\ell$ on the Sobolev scale is any pseudodifferential operator of order $\ell$ acting on $M$.
\end{defi}

\begin{defi}
\label{S_l}
Given $\ell \in \mathbb{R}$. We denote by $\mathcal{S}_{\ell}(M)$ the set consisting of families of operators $\{B_{s}(k)\}_{(s,k)\in \mathbb{S}^{d-1} \times \mathcal{O}}$ acting on $M$ so that the following properties hold:
\begin{itemize}
\item
For any $(s,k) \in \mathbb{S}^{d-1} \times \mathcal{O}$, $B_s(k)$ is of order $\ell$ on the Sobolev scale $(H^{p}(M))_{p \in \mathbb{R}}$.

\item For any $ p \in \mathbb{R}$, the operator $B_s(k)$ is smooth in $k$ as a $\displaystyle B(H^{p}(M), H^{p-\ell}(M))$-valued function. 

\item For any multi-index $\alpha$, $D^{\alpha}_k B_s(k)$ is of order $\ell-|\alpha|$ on the Sobolev scale $(H^{p}(M))_{p \in \mathbb{R}}$ and moreover, for any $p \in \mathbb{R}$, the following uniform condition holds 
$$\sup_{(s,k)\in \mathbb{S}^{d-1} \times \mathcal{O}}\|D^{\alpha}_k B_s(k)\|_{B(H^{p}(M), H^{p-\ell+|\alpha|}(M))}<\infty.$$
\end{itemize}
\end{defi}

It is worth giving a separate definition for the class $\displaystyle \mathcal{S}_{-\infty}(M)=\cap_{\ell \in \mathbb{R}}\mathcal{S}_{\ell}(M)$ as follows:
\begin{defi}
\label{class_S}
We denote by $\mathcal{S}_{-\infty}(M)$ the set consisting of families of smoothing operators $\{U_{s}(k)\}_{(s,k)\in \mathbb{S}^{d-1} \times \mathcal{O}}$ acting on $M$ so that the following properties hold:
\begin{itemize}
\item For any $ m_1, m_2 \in \mathbb{R}$, the operator $U_s(k)$ is smooth in $k$ as a $\displaystyle B(H^{m_1}(M), H^{m_2}(M))$-valued function. 

\item The following uniform condition holds for any multi-index $\alpha$:
$$\sup_{(s,k)\in \mathbb{S}^{d-1} \times \mathcal{O}}\|D^{\alpha}_k U_s(k)\|_{B(H^{m_1}(M), H^{m_2}(M))}<\infty.$$
\end{itemize}
\end{defi}

We now introduce the class $\tilde{S}^{\ell}(\mathbb{T}^n)$ of parameter-dependent toroidal symbols on the $n$-dimensional torus \footnote{Note that for the case $n=d$, the class of parameter-dependent toroidal symbols was introduced in \cite[Definition 7.3]{KKR}. Nevertheless, the techniques and results on parameter-dependent toroidal pseudodifferential operators obtained in \cite[Section 8]{KKR} still work similarly for the general case $n \geq 1$.}. 
\begin{defi}
\label{parameter_class_symbols}
The parameter-dependent class $\tilde{S}^{\ell}(\mathbb{T}^n)$ consists of symbols $\sigma(s,k;x,\xi)$ satisfying the following conditions:
\begin{itemize}
\item For each $(s,k) \in \mathbb{S}^{d-1} \times \mathcal{O}$, the function $\sigma(s,k;\cdot,\cdot)$ is a symbol of order $\ell$ on the torus $\mathbb{T}^n$ (see e.g., \cite[Definition 7.3]{KKR}).

\item Consider any multi-indices $\alpha, \beta, \gamma$ and any $s \in\mathbb{S}^{d-1}$. Then the function $\sigma(s,\cdot;\cdot,\cdot)$ is smooth on $\mathcal{O} \times \mathbb{T}^n \times \mathbb{R}^n$. Furthermore, for some positive constant $C_{\alpha \beta \gamma}$ (independent of $s$,$k$,$x$,$\xi$), we have
\begin{equation*}
\label{E:parameter_class_symbols}
\sup_{s \in \mathbb{S}^{d-1}}|D^{\alpha}_{k}D^{\beta}_{\xi}D^{\gamma}_x \sigma(s,k;x,\xi)| \leq C_{\alpha \beta \gamma} (1+|\xi|)^{m-|\alpha|-|\beta|}.
\end{equation*} 
\end{itemize}

We also define
$$\tilde{S}^{-\infty}(\mathbb{T}^n):=\bigcap_{\ell \in \mathbb{R}}\tilde{S}^{\ell}(\mathbb{T}^n).$$
\end{defi}
The class of pseudodifferential operators on the torus $\mathbb{T}^n$ is also provided in the next definition.
\begin{defi}
\label{periodic_pseudo}
\begin{itemize}
\item
Given a symbol $\sigma(x,\xi)$ of order $\ell$ on the torus $\mathbb{T}^n$,  the corresponding periodic pseudodifferential operator $Op(\sigma)$ is defined by
\begin{equation*}
\label{E:def_pseudo}
\left(Op(\sigma)f\right)(x):=\sum_{\xi \in \mathbb{Z}^n}\sigma(x,\xi)\tilde{f}(\xi)e^{2\pi i \xi \cdot x},
\end{equation*}
where $\tilde{f}(\xi)$ is the Fourier coefficient of $f$ at $\xi$. 

\item For any $\ell \in \mathbb{R} \cup \{-\infty\}$, the set of all families of periodic pseudodifferential operators $\{Op(\sigma(s,k;\cdot,\cdot))\}_{(s,k) \in \mathbb{S}^{d-1} \times \mathcal{O}}$, where $\sigma$ runs over the class $\tilde{S}^{\ell}(\mathbb{T}^n)$, is denoted by $Op(\tilde{S}^{\ell}(\mathbb{T}^n))$.
\end{itemize}
\end{defi}

\begin{remarks}
\label{rem_class}
\begin{enumerate}[(a)]
\item
It is straightforward to check from definitions and the Leibniz rule that for any $\ell_1, \ell_2 \in \mathbb{R} \cup \{-\infty\}$, if $\{A_s(k)\}_{(s,k) \in \mathbb{S}^{d-1} \times \mathcal{O}}$, $\{B_s(k)\}_{(s,k) \in \mathbb{S}^{d-1} \times \mathcal{O}}$ are two families of operators in the class $\mathcal{S}_{\ell_1}(M)$
and $\mathcal{S}_{\ell_2}(M)$, respectively, then the family of operators $\{A_s(k)B_s(k)\}_{(s,k) \in \mathbb{S}^{d-1} \times \mathcal{O}}$ belongs to $\mathcal{S}_{\ell_1+\ell_2}(M)$.

\item If the family of operators $\{B_s(k)\}_{(s,k) \in \mathbb{S}^{d-1} \times \mathcal{O}}$ belongs to the class $\mathcal{S}_{\ell}(M)$ then by definition, the family of operators $\{D^{\alpha}_k B_s(k)\}_{(s,k) \in \mathbb{S}^{d-1} \times \mathcal{O}}$ is in the class $\mathcal{S}_{\ell-|\alpha|}(M)$ for any multi-index $\alpha$.

\item 
$\mathcal{S}_{-\infty}(\mathbb{T}^n)$ is the class $\mathcal{S}$ introduced in \cite[Definition 7.8]{KKR}.

\item
Given a family of symbols $\{\sigma(s,k;\cdot,\cdot)\}_{(s,k) \in \mathbb{S}^{d-1} \times \mathcal{O}} \in \tilde{S}^{\ell}(\mathbb{T}^n)$, it follows from definitions here and boundedness on Sobolev spaces of periodic pseudodifferential operators (see e.g., \cite[Corollary 4.8.3]{Ruzh-Turu}) that the corresponding family of periodic pseudodifferential operators $\{Op(\sigma(s,k;\cdot,\cdot))\}_{(s,k) \in \mathbb{S}^{d-1} \times \mathcal{O}}$ is in the class $\mathcal{S}_{\ell}(\mathbb{T}^n)$. In other words, $Op(\tilde{S}^{\ell}(\mathbb{T}^n)) \subseteq \mathcal{S}_{\ell}(\mathbb{T}^n)$ for any $\ell \in \mathbb{R} \cup \{-\infty\}$.

\end{enumerate}
\end{remarks}

Roughly speaking, the next lemma says that we can deduce regularity of the Schwartz kernel of an operator provided that it acts ``nicely" on Sobolev spaces.
\begin{lemma}
\label{L:kernel_estimate_agmon}
Let $A$ be a bounded operator in $L^2(M)$, where $M$ is a compact $n$-dimensional manifold. Suppose that the range of $A$ is contained in $H^m(M)$, where $m>n/2$ and in addition,
\begin{equation}
\label{sobolev_est}
\|Af\|_{H^m(M)} \leq C\|f\|_{H^{-m}(M)}
\end{equation}
for all $f \in L^2(M)$.

Then $A$ is an integral operator whose kernel $K_A(x,y)$ is a continuous function on $M \times M$. In addition, the kernel of $A$ satisfies the following estimate:
\begin{equation}
\label{E:kernel_estimate_agmon}
|K_A(x,y)| \leq \gamma_0 C,
\end{equation}
where $\gamma_0$ is a constant depending only on $n$ and $m$.
\end{lemma}

\begin{proof}
For the Euclidean case, this fact is shown in \cite[Lemma 2.2]{Ag}. To prove this on a general compact manifold, we simply choose a finite cover $\mathcal{U}=\{U_p\}$ of $M$ with charts $U_p\cong \mathbb{R}^n$. Then fix a smooth partition of unity $\{\varphi_p\}$ with respect to the cover $\mathcal{U}$, i.e., $\supp{\varphi_p} \Subset U_p$. We decompose $A=\sum_{p,q}\varphi_p A \varphi_{q}$. Given any $f \in L^2(M)$, the estimate \mref{sobolev_est} will imply the estimate $\|\varphi_p A\varphi_q f\|_{H^m(U_p)} \leq C\|\varphi_q f\|_{H^{-m}(M)}\leq C\|f\|_{H^{-m}(U_q)}$ for any $p,q$. Hence, we obtain the conclusion of the lemma for the kernel of each operator $\varphi_{p}A\varphi_{q}$, and thus for the kernel of $A$ too.
\end{proof}

In what follows, we will prove a nice behavior of kernels of families of operators in the class $\mathcal{S}_{\ell}(M)$ following from an application of the previous lemma.
\begin{cor}
\label{kernel_in_S}
Assume that $\ell \in \mathbb{R} \cup \{-\infty\}$ and $\{A_s(k)\}_{(s,k)}$ is a family of operators in $\mathcal{S}_{\ell}(M)$. Let $K_{A_s}(k,x,y)$ be the Schwartz kernel of the operator $A_s(k)$. Then for any multi-index $\alpha$ satisfying $|\alpha|\geq n+\ell+2$, the kernel $D^{\alpha}_k K_{A_s}(k,x,y)$ is continuous on $M \times M$ and moreover, the following estimate holds:
$$\sup_{s,k,x,y}|D^{\alpha}_k K_{A_s}(k,x,y)|<\infty.$$
\end{cor}
\begin{proof}
For such $|\alpha|\geq n+\ell+2$, we pick some integer $m \in (n/2, (-\ell+|\alpha|)/2]$. Then by Definition \ref{S_l}, we have
$$\sup_{s,k}\|D^{\alpha}_k A_s(k)f\|_{H^m(M)} \leq C_{\alpha}\|f\|_{H^{-m}(M)}.$$
Applying Lemma \ref{L:kernel_estimate_agmon}, the estimates \mref{E:kernel_estimate_agmon} hold for kernels $D^{\alpha}_k K_{A_s}(k,x,y)$ of the operators $D^{\alpha}_k A_s(k)$ uniformly in $(s,k)$.
\end{proof}

The next theorem shows the inversion formula (i.e., the existence of a family of parametrices) in the case of $\mathbb{T}^n$. The proof of this theorem just comes straight from the proof of \cite[Theorem 8.3]{KKR}. We omit the details.
\begin{thm}
\label{parametrices_euclidean}
Let $r \in \mathbb{N}$. Consider a family of $2r^{th}$ order elliptic operators $\{(\mathcal{Q}_{s}(k)\}_{(s,k) \in \mathbb{S}^{d-1} \times \mathcal{O}}$ on the torus $\mathbb{T}^n$. Assume that this family is in $Op(\tilde{S}^{2r}(\mathbb{T}^n))$ and moreover, for each $(s,k) \in \mathbb{S}^{d-1} \times \mathcal{O}$, the symbol $\sigma(s,k;x,\xi)$ of the operator $\mathcal{Q}_{s}(k)$ is of the form
\begin{equation*}
\label{form_symbols_1}
\sigma(s,k;x,\xi)=L_0(s,k;x,\xi)+\tilde{\sigma}(s,k;x,\xi),
\end{equation*}
where the families of parameter-dependent symbols $\{L_0(s,k;x,\xi)\}_{(s,k)}$, $\{\tilde{\sigma}(s,k;x,\xi)\}_{(s,k)}$ are in the class $\tilde{S}^{2r}(\mathbb{T}^n)$ and $\tilde{S}^{2r-1}(\mathbb{T}^n)$, respectively. Moreover, suppose that there is some constant $A>0$ such that whenever $|\xi|>A$, we have
$$|L_0(s,k;x,\xi)|\geq 1, \quad (s,k,x) \in \mathbb{S}^{d-1}\times \mathcal{O} \times \mathbb{T}^n.$$ 
We call $L_0(s,k;x,\xi)$ the ``leading part" of the symbol $\sigma(s,k;x,\xi)$.

Then there exists a family of parametrices $\{\mathcal{A}_{s}(k)\}_{(s,k)}$ in $Op(\tilde{S}^{-2r}(\mathbb{T}^n))$ such that $$\mathcal{Q}_{s}(k)\mathcal{A}_{s}(k)=I-\mathcal{R}_{s}(k),$$ where $\mathcal{R}_{s}(k)$ is some family of smoothing operators in the class $\mathcal{S}_{-\infty}(\mathbb{T}^n)$.
\end{thm}

To build a family of parametrices on a compact manifold, we will follow closely the strategy in \cite{GUI} by working on open subsets of the torus first and then gluing together to get the final global result.
\begin{thm}
\label{parametrices_manifold}
There exists a family of operators $\{A_s(k)\}_{(s,k) \in \mathbb{S}^{d-1}\times \mathcal{O}}$ in the class $\mathcal{S}_{-2}(M)$ and a family of operators $\{R_s(k)\}_{(s,k) \in \mathbb{S}^{d-1}\times \mathcal{O}}$ in the class $\mathcal{S}_{-\infty}(M)$ such that 
$$(L_s(k)-\lambda)A_s(k)=I-R_s(k).$$
\end{thm}

\begin{proof}
Let $V_p$ ($p=1,\dots,N$) be a finite covering of the compact manifold $M$ by evenly covered coordinate charts. We also choose an open covering $U_p$ ($p=1,\dots,N$) that refine the covering $\{V_p\}$ such that $\overline{U_p} \subset V_p$ for any $p$.
We can assume that each $V_p$ is an open subset of $(0, 2\pi)^n$ in $\mathbb{R}^n$ and hence, we can view each $V_p$ as an open subset of the torus $\mathbb{T}^n$. 

To simplify the notation, we will suppress the index $p=1,\dots,N$ which specifies the open sets $V_p$, $U_p$ until the final steps of the proof.  Let us denote by $i_U$, $r_U$ the inclusion mapping from $i_U: U \rightarrow \mathbb{T}^n$ and the restriction mapping 
$r_U: C^{\infty}(\mathbb{T}^n) \rightarrow C^{\infty}(U)$, correspondingly. We also use the same notation $L_s(k)-\lambda$ for its restrictions to the coordinate charts $V, U$ if no confusion arises. Then $(L_s(k)-\lambda)r_U$ can be considered as an operator on $\mathbb{T}^n$.

Let us first establish the following localized version of the inversion formula
\begin{lemma}
\label{localized}
There are families of symbols $\{a(s,k;x,\xi)\}_{(s,k) \in \mathbb{S}^{d-1}\times \mathcal{O}} \in \tilde{S}^{-2}(\mathbb{T}^n)$ and  
$\{r(s,k;x,\xi)\}_{(s,k) \in \mathbb{S}^{d-1}\times \mathcal{O}} \in \tilde{S}^{-\infty}(\mathbb{T}^n)$ so that
$$(L_s(k)-\lambda)r_U \mathcal{A}_s(k)=r_U (I-\mathcal{R}_s(k)),$$
where $\mathcal{A}_s(k)=Op(a(s,k;\cdot,\cdot))$, $\mathcal{R}_s(k)=Op(r(s,k;\cdot,\cdot))$.
\end{lemma}
\begin{proof}
We denote by $(L_{s}(k)-\lambda)^T$ the transpose operator of $(L_{s}(k)-\lambda)$ on $V$.
Now let $\nu$ be a function in $C^{\infty}_c(V)$ such that $\nu=1$ in a neighborhood of $\overline{U}$ and $0 \leq \nu \leq 1$. Define 
$$\mathcal{Q}_s(k)=(L_s(k)-\lambda)(L_{s}(k)-\lambda)^T\nu+(1-\nu)\Delta^2.$$
Observe that each operator $\mathcal{Q}_s(k)$ is a globally defined $4^{th}$ order differential operator on $\mathbb{T}^n$ with the following principal symbol
$$\nu(x)|\sigma_0(s,k;x,\xi)|^2+(1-\nu(x))|\xi|^4.$$
Here $\sigma_0(s,k;x,\xi)$ is the non-vanishing symbol of the elliptic operator $L_s(k)-\lambda$.
Thus, each operator $\mathcal{Q}_s(k)$ is an elliptic differential operator on $\mathbb{T}^n$.
In order to apply Theorem \ref{parametrices_euclidean} to the family $\{\mathcal{Q}_s(k)\}_{(s,k)}$, we need to study its family of symbols $\{\sigma(s,k;x,\xi)\}_{(s,k)}$.

On the evenly covered chart $V$, we can assume that the operator $L_s(k)-\lambda$ is of the form
$$\sum_{|\alpha| \leq 2}a_{\alpha}(x)(D+(k+i\beta_s)^{T}\cdot\nabla\tilde{h})^{\alpha},$$
for some functions $a_{\alpha} \in C^{\infty}(V)$ and $\tilde{h}$ is a smooth function obtained from the additive function $h$ through some coordinate transformation on the chart $V$.
Similarly, since  $(L_{s}(k)-\lambda)^T=L(k-i\beta_s)-\lambda$, one can write the operator $(L_{s}(k)-\lambda)^T$ on $V$ as follows:
$$\sum_{|\alpha| \leq 2}\tilde{a}_{\alpha}(x)(D+(k-i\beta_s)^{T}\cdot\nabla\tilde{h})^{\alpha},$$
for some functions $\tilde{a}_{\alpha} \in C^{\infty}(V)$. Then, on $\mathbb{T}^n$, the operator $\mathcal{Q}_s(k)$ has the following form:
$$\sum_{|\alpha|, |\beta| \leq 2}a_{\alpha}(x)\tilde{a}_{\beta}(x)(D+(k+i\beta_s)^{T}\cdot\nabla\tilde{h})^{\alpha}(D+(k-i\beta_s)^{T}\cdot\nabla\tilde{h})^{\beta}\nu(x)+(1-\nu(x))\Delta^2.$$

Put
$$L_0^{(1)}(s,k;x,\xi):=\sum_{|\alpha|=2} a_{\alpha}(x) (\xi+(k+i\beta_s)^{T}\cdot\nabla \tilde{h})^{\alpha},$$
$$L_0^{(2)}(s,k;x,\xi):=\sum_{|\beta|=2} \tilde{a}_{\beta}(x)(\xi+(k-i\beta_s)^{T}\cdot\nabla \tilde{h})^{\beta}$$
and
\begin{equation}
\label{form_symbols_2}
L_0(s,k;x,\xi)=\nu(x)L_0^{(1)}(s,k;x,\xi)L_0^{(2)}(s,k;x,\xi)
\\+(1-\nu(x))|\xi|^4.
\end{equation}
Then the symbol $\sigma(s,k;x,\xi)$ of the operator $\mathcal{Q}_{s}(k)$ can be written as
\begin{equation*}
L_0(s,k;x,\xi)+\tilde{\sigma}(s,k;x,\xi),
\end{equation*}
where the family of symbols $\{\tilde{\sigma}(s,k;x,\xi)\}_{(s,k)}$ is in the class $\tilde{S}^{3}(\mathbb{T}^n)$.
Using the boundedness of $\nabla\tilde{h}$ and coefficients $a_{\alpha}$ on the support of $\nu$, we deduce that the family of the symbols of $\{\mathcal{Q}_s(k)\}_{(s,k)}$ is in $\tilde{S}^{4}(\mathbb{T}^n)$.
Thus, our remaining task is to find a constant $A>0$ such that whenever $|\xi|>A$, we obtain $|L_0(s,k;x,\xi)|>1$.
Note that by ellipticity, there are positive constants $\theta_1$, $\theta_2$ such that
$$\sum_{|\alpha|=2} a_{\alpha}(x) \geq \theta_1 |\xi|^2$$
and
$$\sum_{|\alpha|=2} \tilde{a}_{\alpha}(x) \geq \theta_2 |\xi|^2.$$
We define
$$\|a\|_{\infty}:=\sum_{|\alpha|=|\beta|=2}\|a_{\alpha}(\cdot)\|_{L^{\infty}(\supp(\nu))}+\|\tilde{a}_{\beta}(\cdot)\|_{L^{\infty}(\supp(\nu))}$$
and 
$$A_p:=\max_{(s,k,x) \in \mathbb{S}^{d-1} \times \mathcal{O} \times \supp(\nu)}\left(|k^T \cdot \nabla\tilde{h}|^2+\theta_p^{-1}\|a\|_{\infty}|\beta_s^T \cdot \nabla\tilde{h}|^2+\theta_p^{-1}\right), \quad p=1,2.$$ 
Suppose that $\displaystyle |\xi|^2>2\max_{p=1,2}A_p$, then for any $p=1,2$, we have
\begin{align*}
\sqrt{\nu(x)}|L_0^{(p)}(s,k;x,\xi)| &\geq \Re\left(\sqrt{\nu(x)}L_0^{(p)}(s,k;x,\xi)\right) 
\\ & \geq \sqrt{\nu(x)}\left(\theta_p|\xi+k^{T}\cdot\nabla\tilde{h}|^2-\sum_{|\alpha|=2}a_{\alpha}(x)(\beta_s^{T} \cdot \nabla \tilde{h})^{\alpha}\right)
\\ & \geq \sqrt{\nu(x)}\left(\theta_p \left(\frac{|\xi|^2}{2}-|k^T \cdot \nabla\tilde{h}|^2 \right)-\|a\|_{\infty}|\beta_s^T \cdot \nabla\tilde{h}|^2\right)\geq \sqrt{\nu(x)}.
\end{align*}
Thus, due to \mref{form_symbols_2}, if $|\xi|^2>2\max_{p=1,2}A_p+1$ then $|L_0(s,k;x,\xi)|\geq(\sqrt{\nu(x)})^2+(1-\nu(x))|\xi|^4\geq 1$ as we wish.
Now we are able to apply Theorem \ref{parametrices_euclidean} to the family of operators $\{\mathcal{Q}_s(k)\}_{(s,k)}$, i.e., there are families of operators $\{\mathcal{B}_s(k)\}_{(s,k)}\in Op(\tilde{S}^{-4}(\mathbb{T}^n))$ and $\{\mathcal{R}_s(k)\}_{(s,k)}\in \mathcal{S}_{-\infty}(\mathbb{T}^n)$ such that 
$\mathcal{Q}_{s}(k)\mathcal{B}_{s}(k)=I-\mathcal{R}_{s}(k).$

Let $\mathcal{A}_s(k):=(L_s(k)-\lambda)^T \nu \mathcal{B}_s(k)$. Since $\nu=1$ on a neighborhood of $\overline{U}$, we obtain
\begin{equation*}
\begin{split}
r_U (I-\mathcal{R}_s(k))&=r_U \mathcal{Q}_s(k)\mathcal{B}_s(k)\\
&=r_U \left((L_s(k)-\lambda)(L_s(k)-\lambda)^T\nu\mathcal{B}_s(k)+(1-\nu)\Delta^2 \mathcal{B}_s(k)\right)\\
&=r_U (L_s(k)-\lambda)(L_s(k)-\lambda)^T\nu\mathcal{B}_s(k)\\
&=(L_s(k)-\lambda) r_U (L_s(k)-\lambda)^T\nu\mathcal{B}_s(k)=(L_s(k)-\lambda) r_U \mathcal{A}_s(k).
\end{split}
\end{equation*}
In addition, $\{\mathcal{A}_s(k)\}_{(s,k)}\in Op(\tilde{S}^{-2}(\mathbb{T}^n))$ according to the composition formula \cite[Theorem 8.2]{KKR}. Hence, the lemma is proved.
\end{proof}
Let $\mu_p \in C^{\infty}_c(U_p)$ ($p=1,\dots,N$) be a partition of unity with respect to the cover $\{U_p\}_{p=1,\dots,N}$ and for any $p=1,\dots,N$, let $\nu_p$ be a function in $C^{\infty}_c(U_p)$ such that it equals $1$ on a neighborhood of $\supp(\mu_p)$. By Lemma \ref{localized}, there are families of operators $\{\mathcal{A}^{(p)}_s(k)\}_{(s,k)} \in Op(\tilde{S}^{-2}(\mathbb{T}^n))$ and $\{\mathcal{R}^{(p)}_s(k)\}_{(s,k)}\in \mathcal{S}_{-\infty}(\mathbb{T}^n)$ such that
\begin{equation}
\label{parametrices_cover}
(L_s(k)-\lambda)r_{U_p}\mathcal{A}^{(p)}_{s}(k)=r_{U_p}(I-\mathcal{R}^{(p)}_{s}(k)).
\end{equation}
Due to pseudolocality, $(1-\nu_p)\mathcal{A}^{(p)}_s(k)\mu_p \in \mathcal{S}_{-\infty}(\mathbb{T}^n)$. This implies that $r_{U_p}\mathcal{A}^{(p)}_s(k)\mu_p-\nu_p\mathcal{A}^{(p)}_s(k)\mu_p \in \mathcal{S}_{-\infty}(\mathbb{T}^n)$, and thus, 
$$(L_s(k)-\lambda)r_{U_p}\mathcal{A}^{(p)}_s(k)\mu_p-(L_s(k)-\lambda)\nu_p\mathcal{A}^{(p)}_s(k)\mu_p \in \mathcal{S}_{-\infty}(\mathbb{T}^n).$$
By \mref{parametrices_cover}, $\mu_pI-(L_s(k)-\lambda)r_{U_p}\mathcal{A}^{(p)}_{s}(k)\mu_p\in \mathcal{S}_{-\infty}(\mathbb{T}^n).$
Hence, $$\mu_p I-(L_s(k)-\lambda)\nu_p\mathcal{A}^{(p)}_s(k)\mu_p \in \mathcal{S}_{-\infty}(\mathbb{T}^n).$$
Since both operators $\mu_p I$ and $(L_s(k)-\lambda)\nu_p\mathcal{A}^{(p)}_s(k)\mu_p$ are globally defined on the manifold $M$, it follows that 
\begin{equation}
\label{parametrices_cover_2}
\sum_{p} \left(\mu_p I-(L_s(k)-\lambda)\nu_p\mathcal{A}^{(p)}_s(k)\mu_p\right) \in \mathcal{S}_{-\infty}(M).
\end{equation}
Because $Op(\tilde{S}^{-2}(\mathbb{T}^n) \subset \mathcal{S}_{-2}(\mathbb{T}^n)$ (see Remark \ref{rem_class}), each family of operators $\{\mathcal{A}^{(p)}_s(k)\}_{(s,k)}$ is in the class $\mathcal{S}_{-2}(\mathbb{T}^n)$ for every $p$. Since $\{\nu_p\mathcal{A}^{(p)}_s(k)\mu_p\}_{(s,k)}$ is globally defined on $M$, we also have $\{\nu_p\mathcal{A}^{(p)}_s(k)\mu_p\}_{(s,k)} \in \mathcal{S}_{-2}(M)$ for any $p$.

Now define $A_s(k):=\displaystyle \sum_p \nu_p\mathcal{A}^{(p)}_s(k)\mu_p$ and $R_s(k):=I-(L_s(k)-\lambda)A_s(k)$. Then $\{A_s(k)\}_{(s,k)} \in \mathcal{S}_{-2}(M)$ and moreover, due to \mref{parametrices_cover_2}, the family of operators $\{R_s(k)\}_{(s,k)}$ is in $\mathcal{S}_{-\infty}(M)$.   
\end{proof}

The statement of the following lemma is standard.
\begin{lemma}
\label{derivative_ops}
Let $\mathcal{M}$ be a compact metric space, $\mathcal{D}$ be a domain in $\mathbb{R}^{m}$ ($m \in \mathbb{N}$) and $H_1$, $H_2$ be two infinite-dimensional separable Hilbert spaces. Let $\{T_s\}_{s \in \mathcal{M}}$ be a family of smooth maps from $\mathcal{D}$ to $B(H_1, H_2)$ such that for any multi-index $\alpha$, the map $(s,d) \mapsto D^{\alpha}_dT_s(d)$ is continuous from $\mathcal{M} \times \mathcal{D}$ to $B(H_1, H_2)$.
Suppose that there is a family of maps $\{V_s\}_{s \in \mathcal{M}}$ from $\mathcal{D}$ to $B(H_2, H_1)$
such that $V_s(d)T_s(d)=1_{H_1}$ and $T_s(d)V_s(d)=1_{H_2}$ for any $(s,d) \in \mathcal{M} \times \mathcal{D}$. 
Then for each $s \in \mathcal{M}$, the map $d \in \mathcal{D} \mapsto V_s(d)$ is smooth as a $B(H_2, H_1)$-valued function. Furthermore for any multi-index $\alpha$, the map $(s,k) \mapsto D^{\alpha}_d V_s(d)$ is continuous on $\mathcal{M} \times \mathcal{D}$ as a $B(H_2, H_1)$-valued function.
\end{lemma}
%\begin{proof}
%By the assumption, there is a unitary operator $U \in B(H_1, H_2)$ that maps $H_1$ onto $H_2$. We denote by $\tilde{T}_s$, $\tilde{V}_s$ the maps $d \in \mathcal{D} \mapsto U^*T
%_s(d) \in B(H_1)$ and $d \in \mathcal{D} \mapsto V_s(d)U \in B(H_1)$, correspondingly. Then $\tilde{T}_s(d)\tilde{V}_s(d)=\tilde{V}_s(d)\tilde{T}_s(d)=1_{H_1}$.
%Now it suffices to work with the two families of operators $\{\tilde{T}_s\}_{s \in \mathcal{M}}$ and $\{\tilde{V}_s\}_{s \in \mathcal{M}}$.
%Fixing $d_0 \in \mathcal{D}$ and choosing a neighborhood $\mathcal{W} \subset \mathcal{D}$ of $d_0$ such that $\sup_{(s,d) \in \mathcal{M}\times \mathcal{W}}\|\tilde{T}_s(d)-\tilde{T}_s(d_0)\|_{B(H_1)}\cdot \|\tilde{V}_s(d_0)\|_{B(H_1)}<1/2$. For such $d$, the inverse of $\tilde{T}_s(d)$ is given by a Neumann series:
%$$\tilde{V}_s(d)=\tilde{T}_s(d)^{-1}=V_s(d_0)\left(1+\sum_{j \geq 1} ((\tilde{T}_s(d_0)-\tilde{T}_s(d))\tilde{V}_s(d_0))^j\right).$$
%Since the power series in this formula is absolutely convergent on $\mathcal{W}$ (uniformly in $s$), $\tilde{V}_s(d)$ is smooth on $\mathcal{D}$ and each of the derivatives $D^{\alpha}_d \tilde{V}_s(d)$ is also continuous on $\mathcal{M} \times \mathcal{D}$ as a $B(H_1)$-valued function. Hence, the statements for $V_s(d)$ follow immediately.
%\end{proof}

We now go back to the family of operators $\{T_s(k)\}_{(s,k) \in \mathbb{S}^{d-1} \times \mathcal{O}}$. The next statement is the main ingredient in establishing Proposition \ref{kernel_estimate}.
\begin{prop}
\label{P:pseudo_ops}
There is a family of operators $\{B_{s}(k)\}_{(s,k)}$ in $\mathcal{S}_{-2}(M)$ such that the family of operators $\{T_s(k)-B_s(k)\}_{(s,k)}$ belongs to $\mathcal{S}_{-\infty}(M)$.
\end{prop}

\begin{proof}
Due to Theorem \ref{parametrices_manifold}, we can find a family $\{A_s(k)\}_{(s,k)} \in \mathcal{S}_{-2}(M)$ and a family $\{R_s(k)\}_{(s,k)} \in \mathcal{S}_{-\infty}(M)$ such that 
$$(L_s(k)-\lambda)A_s(k)=I-R_s(k).$$
Also, from the definition of $T_s(k)$, we obtain $T_s(k)(L_s(k)-\lambda)=I-\eta(k)P(k+i\beta_s)$. 

Using the above two equalities, we obtain
\begin{equation*}
\label{E:eqn_T(k)}
T_s(k)=A_s(k)-\eta(k)P(k+i\beta_s)A_s(k)+T_s(k)R_s(k).
\end{equation*}

We recall from Section 4 that $P(k+i\beta_s)$ projects $L^2(M)$ onto the eigenspace spanned by the eigenfunction $\phi_{k+i\beta_s}$. Hence, its kernel is the following function
$$\frac{\phi_{k+i\beta_s}(x)\overline{\phi_{k-i\beta_s}(y)}}{(\phi_{k+i\beta_s},\phi_{k-i\beta_s})_{L^2(M)}},$$
which is smooth due to Lemma \ref{L:joint_continuity}.
Thus, the family of operators $\{\eta(k)P(k+i\beta_s)\}_{(s,k)}$ is in $\mathcal{S}_{-\infty}(M)$. Also, the family of operators $\{\eta(k)Q(k+i\beta_s)\}_{(s,k)}$ belongs to $\mathcal{S}_{0}(M)$.

We put $B_s(k):=A_s(k)-\eta(k)P(k+i\beta_s)A_s(k)$, then $\{B_s(k)\}_{(s,k)} \in \mathcal{S}_{-2}(M)$. Since $T_s(k)-B_s(k)=T_s(k)R_s(k)$, the remaining task is to check that the family of operators $\{T_s(k)R_s(k)\}_{(s,k)}$ belongs to the class $\mathcal{S}_{-\infty}(M)$. 

Let us consider any two real numbers $m_1$ and $m_2$. By Lemma \ref{fiber_op_structure}, the operators $L_s(k)-\lambda$ and $L_s(k)Q(k+i\beta_s)-\lambda$ are smooth in $k$ as $B(H^{m_2}(M), H^{m_2-2}(M))$-valued functions such that their derivatives with respect to $k$ are jointly continuous in $(s,k)$.
On the other hand, we can rewrite (see  \cite[Lemma 7.7]{KKR}):
$$T_s(k)=(1-\eta(k))(L_s(k)-\lambda)^{-1}+\eta(k)\lambda^{-1}P(k+i\beta_s)+\eta(k)(L_s(k)Q(k+i\beta_s)-\lambda)^{-1}.$$  
Hence, by Lemma \ref{derivative_ops}, 
%\footnote{Recall that for each real number $p$, the Sobolev space $H^p(M)$ is an infinite-dimensional separable Hilbert space.} 
$T_s(k)$ is smooth in $k$ as a $B(H^{m_2-2}(M), H^{m_2}(M))$-valued function and its derivatives with respect to $k$ are jointly continuous in $(s,k)$. Therefore, for any multi-index $\alpha$, we have
$$\sup_{(s,k) \in \mathbb{S}^{d-1} \times \mathcal{O}}\|D^{\alpha}_k T_s(k)\|_{B(H^{m_2-2}(M), H^{m_2}(M))}<\infty.$$
Moreover since $\{R_s(k)\}_{(s,k)} \in \mathcal{S}_{-\infty}(M)$, $R_s(k)$ is smooth as a $B(H^{m_1}(M), H^{m_2-2}(M))$-valued function and for any multi-index $\alpha$, 
$$\sup_{(s,k)\in \mathbb{S}^{d-1} \times \mathcal{O}} \|D^{\alpha}_k R_s(k)\|_{B(H^{m_1}(M), H^{m_2-2}(M))}<\infty.$$
One can deduce from the Leibniz rule that the composition $T_s(k)R_s(k)$ is smooth in $k$ as a $B(H^{m_1}(M), H^{m_2}(M))$-valued function and for any multi-index $\alpha$, the following uniform condition also holds
$$\sup_{(s,k)\in \mathbb{S}^{d-1} \times \mathcal{O}} \|D^{\alpha}_k (T_s(k)R_s(k))\|_{B(H^{m_1}(M), H^{m_2}(M))}<\infty.$$
Consequently, $\{T_s(k)R_s(k)\}_{(s,k)\in \mathbb{S}^{d-1} \times \mathcal{O}}\in\mathcal{S}_{-\infty}(M)$ as we wish.
\end{proof}
We now finish this subsection.
\begin{prmicrolocalprop}
Proposition \ref{P:pseudo_ops} provides us with the decomposition $T_s(k)=B_s(k)+C_s(k)$, where $\{B_{s}(k)\}_{(s,k)}\in \mathcal{S}_{-2}(M)$ and $\{C_s(k)\}_{(s,k)} \in \mathcal{S}_{-\infty}(M)$. Let $K_{B_s}(k,x,y)$, $K_{C_s}(k,x,y)$ be the Schwartz kernels of $B_s(k)$ and $C_s(k)$, correspondingly. It follows from Corollary \ref{kernel_in_S} that for any multi-index $\alpha$ satisfying $|\alpha| \geq n$, each kernel $D^{\alpha}_k K_{B_s}(k,x,y)$ is continuous on $M \times M$. Furthermore, we have
$$\sup_{(s,k,x,y) \in \mathbb{S}^{d-1} \times \mathcal{O} \times M \times M}|D^{\alpha}_k K_{B_s}(k,x,y)|<\infty.$$
A similar conclusion also holds for the family of kernels $\{K_{C_s}(k,x,y)\}_{(s,k)}$ and thus, for the family of kernels $\{K_{s}(k,x,y)\}_{(s,k)}$ too. This finishes the argument.
\end{prmicrolocalprop}

%%%%%%%%%%%%%%%%%%%%%%%%%%%%%%%%%%%%%%%%%%%%%%%
\section{Asymptotics of Green's functions and Martin compactifications for nonsymmetric second-order periodic elliptic operators}
In this section, we discuss briefly analogous results for nonsymmetric $G$-equivariant\footnote{Here, without loss of generality, we assume that $G=\mathbb{Z}^d$.} second-order elliptic operators on an abelian covering $X$ below and at the generalized principal eigenvalue, which generalize the main results in \cite{MT}. Let us consider 
now a $G$-periodic linear elliptic operator $A$ of second-order acting on $C^{\infty}(X)$ such that in any coordinate system $(U; x_1, \cdots, x_n)$, $A$ has the form:
$$A=-\sum_{1 \leq i,j \leq n}a_{ij}(x)\partial_{x_i}\partial_{x_j}+\sum_{1 \leq i \leq n}b_i(x)\partial_{x_i}+c(x),$$
where $a_{ij}, b_i, c$ are smooth, real-valued, periodic functions. The matrix $a(x):=(a_{ij}(x))_{1 \leq i,j \leq n}$ is positive definite. The generalized principal eigenvalue of $A$ is defined as follows
$$\Lambda_A=\sup\{\lambda \in \mathbb{R} \mid Au=\lambda u \hspace{4pt} \mbox{for some positive solution} \hspace{4pt} u\}.$$
Let $A^*$ be the formal adjoint operator of $A$. The generalized principal eigenvalues of $A^*$ and $A$ are equal, i.e., $\Lambda_A=\Lambda_{A^*}$. Also, the operators $A-\lambda$ and $A^*-\lambda$ are subcritical \footnote{I.e., positive Green's functions exist for these operators.} if $\lambda<\Lambda_{A}$.

Recall that a function $u$ on $X$ is called a $G$-multiplicative function with exponent $k \in \mathbb{R}^d$ if $u(g \cdot x)=e^{k \cdot g}u(x), \forall x \in X, g \in G$ (see Definition \ref{add_mul}). For any $k \in \mathbb{R}^d$, it is known from \cite{Agmon, LinPinchover} that there exists a unique \textit{real} number $\Lambda_A(k)$ so that the equation $Au=\Lambda_A(k)u$ admits a positive $G$-multiplicative solution $u$ with exponent $-k$. 

We extract the following results from \cite{Agmon, LinPinchover, Pinsky}.
\begin{thm}
\label{AP}
\begin{enumerate}[(a)]
\item
Let $h$ be an additive function on $X$.
Then $\Lambda_A(k)$ is the principal eigenvalue of $A(ik)$ with multiplicity one, where $A(ik)$ is the operator $e^{k\cdot h(x)}Ae^{-k\cdot h(x)}$. Furthermore, $\Lambda_A=\max_{k \in \mathbb{R}^d}\Lambda_A(k)=\Lambda_A(\beta_0)$ for a unique $\beta_0 \in \mathbb{R}^d$. 
\item
The function $\Lambda_A(k)$ is real analytic,  strictly concave, bounded from above, and its gradient vanishes at only its maximum point $k=\beta_0$. The Hessian of the function $\Lambda_A(k)$ is negative definite at any $k \in \mathbb{R}^d$.
\item
For any $\lambda \in \mathbb{R}$, we define 
\begin{equation}
\label{levelsets}
\begin{split}
K_{\lambda}&=\{k \in \mathbb{R}^d \mid \Lambda_A(k)\geq \lambda\},\\
\Gamma_{\lambda}&=\partial K_{\lambda}=\{k \in \mathbb{R}^d \mid \Lambda_A(k)=\lambda\}.
\end{split}
\end{equation}
Then $\Gamma_{\lambda}$ (resp. $K_{\lambda}$) is the set consisting of all vectors $k \in \mathbb{R}^d$ such that $Au=\lambda u$ (resp. $Au \geq \lambda u$) for some positive $G$-multiplicative function $u$ with exponent $-k$. 
Moreover, if $\lambda=\Lambda_A$, $\Gamma_{\lambda}=K_{\lambda}=\{\beta_0\}$ while if $\lambda<\Lambda_{A}$, $K_{\lambda}$ is a $d$-dimensional strictly convex compact subset in $\mathbb{R}^d$ and its boundary $\Gamma_{\lambda}$ is a compact $d-1$ dimensional analytic submanifold of $\mathbb{R}^d$. 
In all cases, $\Gamma_{\lambda}$ is the set of all extreme points of $K_{\lambda}$.
\item
We define analogous level sets $K^*_{\lambda}$, $\Gamma^*_{\lambda}$ as in \mref{levelsets} for the formal adjoint operator $A^*$, . Then $K^*_{\lambda}=-K_{\lambda}$ and $\Gamma^*_{\lambda}=-\Gamma_{\lambda}$.
Also, $\Lambda_{A}(k)=\Lambda_{A^*}(-k)$ for any $k \in \mathbb{R}^d$. In particular, if $A=A^*$, $\Lambda_{A}(k)$ is an even function and $K_{\Lambda_A}=\{0\}$ (or $\beta_0=0$).
\end{enumerate}
\end{thm}
We are interested in finding the asymptotics at infinity of the Green's function $G_{\lambda}(x,y)$ of the operator $A-\lambda$, where $\lambda \in (-\infty, \Lambda_A]$. 
From now on, we fix a point $x_0$ in $X$. Let $\mathcal{K}_{A,\lambda}$ be the set consisting of all positive solutions $u$ of the equation $Au=\lambda u$ such that $u$ is normalized at $x_0$, i.e., $u(x_0)=1$. We denote by $\mathcal{M}_{A, \lambda}$ the subset of $\mathcal{K}_{A,\lambda}$ containing all normalized (at $x_0$) positive $G$-multiplicative solutions with exponents in $\Gamma_{\lambda}$. It was proved in \cite{Agmon, LinPinchover} that $\mathcal{M}_{A, \lambda}$ coincides with the set of all extreme points of the convex compact set $\mathcal{K}_{A, \lambda}$ \footnote{This result also holds for $G$-equivariant elliptic operators of second-order on a Riemannian co-compact nilpotent covering $X$, see e.g., \cite[Theorem 6.8]{LinPinchover}.}. As a consequence of the latter fact, it turns out that all such positive $G$-multiplicative solutions are exactly all minimal positive solutions of the equation $(A-\lambda)u=0$.
When $\lambda$ is below the generalized principal eigenvalue $\Lambda_A$,
Theorem \ref{AP} (compare to Proposition \ref{analytic_perturbation} and Lemma \ref{func_E}) enables us to define the following notions in a similar manner to the discussion in Subsection \ref{setup2.3} \footnote{The role of the function $E$ in Subsection \ref{setup2.3} is now played by the function $\Lambda_A$.}.
For each $s \in \mathbb{S}^{d-1}$, let $\beta_s$ be the unique point in $\Gamma_{\lambda}$ such that 
$$\frac{\nabla \Lambda_A(\beta_s)}{|\nabla \Lambda_A(\beta_s)|}=-s.$$
For any $k \in \mathbb{R}^d$, let $\phi_k$ and $\phi^*_k$ be periodic, positive, and normalized\footnote{Here we mean that $\phi_k(x_0)=\phi^*_k(x_0)=1$.} solutions of the equations $A(ik)u=\Lambda_A(k) u$, $A(ik)^*u^*=\Lambda_{A}(k) u^*$, respectively.
We have:
\begin{thm}
\label{main-nsa} 
\begin{enumerate}[(a)]
\item
Suppose that $d \geq 2$ and $\lambda<\Lambda_A$. Then as $d_{X}(x,y) \rightarrow \infty$,  the following asymptotics of the Green's function $G_{\lambda}$ of $A-\lambda$ holds:
\begin{equation*}
\label{main_asymp_nsa}
\begin{split}
G_{\lambda}(x,y)&=\frac{e^{-(h(x)-h(y))\cdot\beta_{s}}}{(2\pi|h(x)-h(y)|)^{(d-1)/2}}\cdot\frac{|\nabla \Lambda_A(\beta_s)|^{(d-3)/2}}{\det{(-\mathcal{P}_s \Hess{(\Lambda_A)}(\beta_{s})\mathcal{P}_s)}^{1/2}} \cdot\frac{\phi_{\beta_{s}}(x)\phi^*_{\beta_{s}}(y)}{(\phi_{\beta_{s}},\phi^*_{\beta_{s}})_{L^{2}(M)}}\\&
+e^{(h(y)-h(x))\cdot \beta_{s}}O(d_X(x,y)^{-d/2}),
\end{split}
\end{equation*}
where $\displaystyle s=(h(x)-h(y))/|h(x)-h(y)| \in \mathcal{A}_h$ 
and $\mathcal{P}_s$ is the same projection we defined in Theorem \ref{main}.
\item
Suppose that $d \geq 3$ and $\Lambda_A=\Lambda_A(\beta_0)$. Then the minimal Green's function $G(x,y)$ of $A-\Lambda_A$ admits the following asymptotics as $d_X(x,y) \rightarrow \infty$:
\begin{equation*}
\label{main_asymp_KR_nsa}
\begin{split}
G(x,y)=\frac{\Gamma(\frac{d-2}{2})e^{-(h(x)-h(y))\cdot \beta_0}}{2\pi^{d/2}\sqrt{\det H}|H^{-1/2}(h(x)-h(y))|^{d-2}}\cdot \frac{\phi_{\beta_0}(x)\phi^*_{\beta_0}(y)}{(\phi_{\beta_0}, \phi^*_{\beta_0})_{L^2(M)}}\cdot\left(1+O\left(d_X(x,y)^{-1}\right)\right),
\end{split}
\end{equation*}
where $H=-\Hess(\Lambda_A)(\beta_0)$.
\end{enumerate}
\end{thm}

\begin{proof}
\begin{enumerate}[(a)]
\item
By considering the operator $e^{\beta_0 \cdot h(x)}Ae^{-\beta_0 \cdot h(x)}$ instead of $A$, we can assume that $\beta_0=0$. 
We follow the proof of Theorem \ref{main} (see also the outline of the proof at the end of Section \ref{mainresults}). To apply the Floquet reduction of the problem as in Subsection \ref{fl-reduction}, we need to obtain the following analog of Proposition \ref{singularity}: for any $t \in [0,1]$ and $k \in \mathbb{R}^d$, then $\lambda$ is in the resolvent set of $A(k+it\beta_s)$ if and only if $k \in 2\pi\mathbb{Z}^d$ and $t=1$. Indeed, if $t=1$, this statement follows directly from \cite[Lemma 5.8]{KP2} (see also \cite[Lemma 15]{KP1}). Otherwise, when $t\in [0,1)$, then by the concavity of $\Lambda_A$ (Theorem \ref{AP}), one has $\Lambda(t\beta_s)\geq t\Lambda(\beta_s)+(1-t)\Lambda(0)>\lambda$. This allows us to apply \cite[Lemma 5.8]{KP2} again to conclude that $\lambda$ is not in $\sigma(A(k+it\beta_s))$ for any $k \in \mathbb{R}^d$. By using this fact, for any $s \in \mathcal{A}_h$, the integral kernel $G_{s,\lambda}(x,y)$ of the operator $R_{s,\lambda}$ defined via \mref{E:R_s} exists (see Lemma \ref{L:bilin}). Now we can repeat the argument \footnote{Note that the proof of Theorem \ref{microlocal} also works in this case.} in Subsection \ref{isolate-green} to see that the asymptotics of $G_{s,\lambda}$ is the same as the asymptotics of the following integral (see \mref{E:formula_G0}): 
$$G_0(x,y)=\frac{1}{(2\pi)^d}\int_{[-\pi, \pi]^d}\frac{e^{ik\cdot (h(x)-h(y))}\eta(k)}{\Lambda_A(\beta_s-ik)-\lambda}\cdot\frac{\phi_{k+i\beta_s}(x)\phi^*_{k+i\beta_s}(y)}{(\phi_{k+i\beta_s}, \phi^*_{k+i\beta_s})_{L^2(M)}}dk,$$
where $\eta$ is a cut-off smooth function on $(-\pi, \pi)^d$ such that $\eta(k)=1$ around $k=0$ and moreover, the function $\Lambda_A(\cdot)$ has an analytic continuation to an open neighborhood of the support of $\eta$. To finish the proof, we just use Proposition \ref{P:asymp_KKR} to find the asymptotics of $G_0$ like in the proof of Theorem \ref{main}.
\item
This is similar to the proof of Theorem \ref{main_KR}. We skip the details.
\end{enumerate}
\end{proof}
Note that when $A$ is symmetric and $\lambda$ does not exceed the bottom of $\sigma(A)$, the quasimomentum $k_0$ in Assumption A is zero due to Theorem \ref{AP} (d), and thus, the asymptotics in Theorem \ref{main} and Theorem \ref{main-nsa} are the same one.

As an application of Theorem \ref{main-nsa}, we describe the Martin compactifications, Martin boundaries, and Martin integral representation for such operators (see e.g., \cite{Murata-skew, Pinsky, Pinsky_bk, Woess} for some basic background on Martin boundary theory). This also generalizes Theorem 1.5 and Theorem 1.7 in \cite{MT}.

\begin{thm}
\label{martin}
\begin{enumerate}[(a)]
\item
Let $d \geq 2$ and $\lambda \in (-\infty, \Lambda_A)$. Then both of the Martin boundary and the minimal Martin boundary of the abelian covering $X$ for the operator $A-\lambda$ are homeomorphic to $\Gamma_{\lambda}$ (or the sphere $\mathbb{S}^{d-1}$). The Martin compactification of $X$ for $A$ is equal to $X \cup \mathbb{S}^{d-1}$, i.e., $X$ is adjoined by the sphere $\mathbb{S}^{d-1}$ at infinity. Furthermore, for any normalized positive solution $u$ of the equation $Au=\lambda u$ (i.e., $u \in \mathcal{K}_{A,\lambda}$), there exists a unique regular Borel probability measure $\mu$ on $\mathbb{S}^{d-1}$ such that
$$u(x)=\int_{\mathbb{S}^{d-1}}e^{-(h(x)-h(x_0))\cdot \beta_s}\phi_{\beta_s}(x)d\mu(s).$$
\item
Let $d \geq 3$ and $\Lambda_A=\Lambda_A(\beta_0)$. Then the Martin boundary and the minimal Martin boundary coincide with the set $\Gamma_{\Lambda_A}=\{\beta_0\}$. The Martin compactification is the one-point compactification of $X$. Moreover, any positive solution $u$ of the equation $Au=\Lambda_A u$ is a positive scalar multiple of the function $e^{- h(x)\cdot \beta_0}\phi_{\beta_0}(x)$.
\end{enumerate}
\end{thm}
\begin{proof}
\begin{enumerate}[(a)]
\item
For any $x,y$ in $X$, we define
\begin{equation*}
\label{martin-kernel}
\begin{split}
K_{\lambda}(x,y)&=\frac{G_{\lambda}(x,y)}{G_{\lambda}(x_0,y)}, \quad y \neq x_0,\\
K_{\lambda}(x,y)&=\delta_{x,x_0}, \quad y=x_0.
\end{split}
\end{equation*}
Let us denote by $(\partial_{M}X)_{A-\lambda}$ the Martin boundary for $A-\lambda$ on $X$. Then $(\partial_{M}X)_{A-\lambda}$ consists of all equivalent classes of fundamental sequences
$\{y_m\}_m$ in $X$. Here $\{y_m\}_m$ is called fundamental if it has no accumulation point in $X$ and the sequence $\{K_{\lambda}(\cdot, y_m)\}_m$ converges uniformly on any compact subset of $X$ to a positive solution of the equation $Au=\lambda u$. Two fundamental sequences $\{y_m\}_m$ and $\{z_m\}_m$ are equivalent if on any compact subset of $X$, $\{|K_{\lambda}(\cdot, y_m)-K_{\lambda}(\cdot, z_m)|\}_m$ converges uniformly to zero. 

Consider a fundamental sequence $\{y_m\}_m$ in $(\partial_{M}X)_{A-\lambda}$.
Then there exists a subsequence $\{y_{m_k}\}_k$ such that $\{h(y_{m_k})/|h(y_{m_k})|\}_k$ converges to a unit vector $s \in \mathbb{S}^{d-1}$ and $\displaystyle \lim_{k \rightarrow \infty} |h(y_{m_k})|=\infty$. 
By Theorem \ref{main-nsa} (a), we have
\begin{equation}
\label{martin-kernel-rhs}
\lim_{k \rightarrow \infty}K_{\lambda}(x,y_{m_k})=e^{-(h(x)-h(x_0))\cdot \beta_s} \frac{\phi_{\beta_s}(x)}{\phi_{\beta_s}(x_0)}=e^{-(h(x)-h(x_0))\cdot \beta_s}\phi_{\beta_s}(x).
\end{equation}
If we denote by $K_{\lambda}(x,s)$ the right-hand side of \mref{martin-kernel-rhs}, then $K_{\lambda}(\cdot,s)$ is a (minimal) positive solution in $\mathcal{M}_{A,\lambda}$. 
Also, $K_{\lambda}(x,y_{m_k}) \rightarrow K_{\lambda}(x,s)$ uniformly on any compact subset of $X$. Since $K(\cdot,s_1) \neq K_{\lambda}(\cdot, s_2)$ if $s_1 \neq s_2$ in $\mathbb{S}^{d-1}$, we must have $\displaystyle \lim_{m \rightarrow \infty} |h(y_{m})|=\infty$ and $\displaystyle\lim_{m \rightarrow \infty}h(y_{m})/|h(y_{m})|=s$. 
This implies that 
\begin{equation}
\label{martin-boundary}
(\partial_{M}X)_{A-\lambda}=\left\{s \in \mathbb{S}^{d-1} \mid \exists \hspace{3pt} \{y_m\}_m \subset X \hspace{3pt} \mbox{such that} \hspace{3pt} |h(y_m)| \rightarrow \infty, \hspace{3pt} \frac{h(y_m)}{|h(y_m)|} \rightarrow s\right\}.
\end{equation}
The right-hand side of \mref{martin-boundary} coincides with the closure of $\mathcal{A}_h$, which is $\mathbb{S}^{d-1}$ by Proposition \ref{A_h}. This proves the first statement. The latter statement follows from the Martin integral representation theorem (see e.g., \cite[Theorem 1.1]{Murata-skew}).
\item
In this case, by Theorem \ref{main-nsa} (b), the Martin kernel is equal to 
$$\lim_{d_X(y,x_0) \rightarrow \infty}\frac{G(x,y)}{G(x_0,y)}=e^{-(h(x)-h(x_0))\cdot \beta_0}\phi_{\beta_0}(x).$$ 
This proves the statement immediately.
\end{enumerate}
\end{proof}
We remark that the integral representation type results stated in Theorem \ref{martin} are special cases of \cite[Theorem 6.11]{LinPinchover}, which also holds for periodic elliptic operators of second-order on nilpotent Riemannian co-compact coverings.
%%%%%%%%%%%%%%%%%%%%%%%
\section{Concluding remarks}
\begin{itemize}
\item 
Notice that Theorem \ref{main_KR} (see also \cite[Theorem 2]{KR}) can be applied to operators with periodic magnetic potentials since this result does not require the realness of the operator $L$. On the other hand, the conditions that $L$ has real coefficients and the gap edge occurs at a high symmetry point of the Brillouin zone (i.e., the assumption \textbf{A5}) are assumed in the inside-the-gap situation, mainly because the central symmetry of the relevant dispersion branch $\lambda_j(k)$ (see Lemma \ref{func_E}) is needed for the formulation and the proof of Theorem \ref{main} (see also \cite[Theorem 2.11]{KKR}).

\item 
The main results in this paper can be easily carried over to the case when the band edge occurs at finitely many quasimomenta $k_0$ in the Brillouin zone (instead of assuming the condition \textbf{A3}) by summing the asymptotics coming from all these non-degenerate isolated extrema.\\
It was shown in \cite{FilKach} that for a wide class of two dimensional periodic second-order elliptic operators (including the class of operators we consider in this paper and periodic magnetic Schr\"odinger operators in $2D$), the extrema of any spectral band function (not necessarily spectral edges) are attained on a finite set of values of the quasimomentum in the Brillouin zone. 

\item The proofs of the main results go through verbatim for periodic elliptic second-order operators acting on vector bundles over the abelian covering $X$. 
\end{itemize}

%%%%%%%%%%%%%%%%%
\section{Acknowledgements}
The work of the author was partly supported by the NSF grant DMS-1517938. The author is grateful to his advisor, Professor  P. Kuchment, for insightful discussion and helpful comments. He also thanks the reviewer for giving constructive comments on this manuscript.
%%%%%%%%%%%%%%%%%%

\end{document}